\documentclass[12pt]{birkjour}

\usepackage{graphicx}
\usepackage{amssymb}
\usepackage{epstopdf,comment}
\usepackage{amsthm,  mathrsfs, enumerate}

\usepackage{amsmath,color}
\usepackage{MnSymbol}

\newtheorem{theorem}{Theorem}
\newtheorem{corollary}{Corollary}
\newtheorem{lemma}{Lemma}
\newtheorem{assumption}{Assumption}

\numberwithin{theorem}{section}
\numberwithin{lemma}{section}
\numberwithin{equation}{section}
\numberwithin{proposition}{section}
\numberwithin{corollary}{section}
\numberwithin{assumption}{section}

\newcommand{\req}[1]{Eq.\,(\ref{#1})}

\begin{document} 
\title[A homogenization theorem for Langevin systems ]{A homogenization theorem for Langevin systems with an application to Hamiltonian dynamics}

\author[Birrell]{Jeremiah Birrell}
\address{Department of Mathematics\\
University of Arizona\\
Tucson, AZ, 85721, USA}
\email{jbirrell@math.arizona.edu}

\author[Wehr]{Jan Wehr}
\address{Department of Mathematics and Program in Applied Mathematics\\
University of Arizona\\
Tucson, AZ, 85721, USA}
\email{wehr@math.arizona.edu}



\date{\today}

\begin{abstract}
This paper studies homogenization of stochastic differential systems.  The standard example of this phenomenon is the small mass limit of Hamiltonian systems.  We consider this case first from the heuristic point of view, stressing the role of detailed balance and presenting the heuristics based on a multiscale expansion.  This is used to propose a physical interpretation of recent results by the authors, as well as to motivate a new theorem proven here.  Its main content is a sufficient condition, expressed in terms of solvability of an associated partial differential equation (``the cell problem''), under which the homogenization limit of an SDE is calculated explicitly.  The general theorem is applied to a class of systems, satisfying a generalized detailed balance condition with a position-dependent temperature.
\end{abstract}
\maketitle

\section{Introduction and background}

This paper studies the small mass limit of a general class of Langevin equations.  Langevin dynamics is defined in terms of canonical variables---positions and momenta---by adding damping and (It\^o) noise terms to Hamiltonian equations.  In the limit when the mass, or masses, of the system's particles, go to zero, the momenta homogenize, and one obtains a limiting equation for the position variables only.  This is a great simplification which often allows one to see the nature of the dynamics more clearly.  If the damping matrix of the original system depends on its state, a {\it noise-induced drift} arises in the limit.  We analyze and interpret this term from several points of view.  The paper consists of four parts.  The first part contains general background on stochastic differential equations.  In the second part, the small-mass limit of Langevin equations is studied using a multiscale expansion.  This method requires making additional assumptions, but it leads to correct results in all cases in which rigorous proofs are known. The third part presents a new rigorous result about homogenization of a general class of singularly perturbed SDEs. The final part applies this result to prove a theorem about the homogenization of a large class of Langevin systems.

\subsection{Stochastic differential equations}

Let us start from a general background on Langevin equations.  The material presented here is not new, and its various versions can be found in many textbooks, see for example \cite{khasminskii2011stochastic}. We do not strive for complete  precision or a  listing of all necessary assumptions in our discussions here. The aim of the first two sections is to motivate and facilitate reading the remainder of the paper. Detailed technical considerations will be reserved for Sections 3 and 4, where we present our new results.

Consider the stochastic differential equation 
\begin{align}\label{SDE}
dy_t = b(y_t)\,dt + \sigma(y_t)\,dW_t.
\end{align}
The process $y_t$ takes values in $\mathbb{R}^m$, $b$ is a vector field in $\mathbb{R}^m$, $W$ is an $n$-dimensional Wiener process and $\sigma$ is an $m \times n$-matrix-valued function.  Define an $m \times m$ matrix $\Sigma$ by $\Sigma =  \sigma\sigma^T$.  The equation \req{SDE} defines a Markov process with the infinitesimal operator 
\begin{align}\label{generator}
(Lf)(y) = {1 \over 2}\Sigma_{ij}\partial_i\partial_j f + b_i\partial_i f    
\end{align}
where we are writing $\partial_i$ for $\partial_{y_i}$ and suppressing the dependence of $\Sigma$, $b_i$ and $f$ on $y$ from the notation.  Summation over repeating indices is implied.  

We assume that this process has a unique stationary probability measure with a $C^2$-density $h(y)$.   Under this assumption $h$ satisfies the equation
\begin{align} \label{stationaryFP}
L^*h = 0
\end{align}
where $L^*$ denotes the formal adjoint of $L$,
\begin{align} \label{adjoint_generator}
L^*f = {1 \over 2}\partial_i\partial_j\left(\Sigma_{ij}f\right) - \partial_i\left(b_i f\right).
\end{align}
That is, we have 
\begin{align}\label{stationaryFP_explicit}
\partial_i\left({1 \over 2}\partial_j\left(\Sigma_{ij}h\right) - b_ih\right) = 0.
\end{align}

Consider the special case when $h$ solves the equation
\begin{align}\label{inner}
{1 \over 2}\partial_j\left(\Sigma_{ij}h\right) - b_ih = 0 .
\end{align}
In this case the operator $L$ is symmetric on the space $L^2\left(\mathbb{R}^m, h\right)\equiv L^2_h$ of square-integrable functions with the weight $h$, as the following calculation shows.  Using product formula, we have
\begin{equation}
\int\left(Lf\right)gh = \int fL^*\left(gh\right) = \int f\partial_i\left({1 \over 2}\partial_j\left(\Sigma_{ij}gh\right) - b_igh\right).
\end{equation}
The expression in parentheses equals
\begin{equation}
{1 \over 2}\partial_j g \left(\Sigma_{ij}h\right) + g{1 \over 2}\partial_j\left(\Sigma_{ij}h\right) - gb_ih = {1 \over 2}\partial_j g \left(\Sigma_{ij}h\right)
\end{equation}
by \req{inner}.  Applying product formula again, we obtain
\begin{equation}
\int\left(Lf\right)gh = \int f\left({1 \over 2}\Sigma_{ij}\left(\partial_i\partial_jg\right)h + {1 \over 2}\partial_i\left(\Sigma_{ij}h\right)\partial_jg\right)
\end{equation}
which, by another application of \req{inner}, equals
\begin{equation}
\int f\left({1 \over 2}\Sigma_{ij}\partial_i\partial_jg - b_j\partial_jg\right)h = \int f\left(Lg\right)h.
\end{equation}
Here is a more complete discussion:
\subsection{Detailed balance condition and symmetry of the infinitesimal operator}
We have
\begin{align}
\int\left(Lf\right)gh &= \int\left(\frac{1}{2}\Sigma_{ij}\partial_i\partial_jf + b_i\partial_if\right)gh \\
&=-\frac{1}{2}\int\partial_if\partial_j\left(\Sigma_{ij}gh\right) + \int\left(\partial_if\right)b_igh \notag\\
&= -\frac{1}{2}\int\partial_if\left[\partial_j\left(\Sigma_{ij}h\right)g + \Sigma_{ij}h\partial_jg\right] + \partial\left(\partial_if\right)b_igh \notag\\
&=\int\partial_if\left[-\frac{1}{2}\partial_j\left(\Sigma_{ij}h\right) + b_ih\right]g - \frac{1}{2}\int \Sigma_{ij}\partial_if\partial_jgh. \notag
\end{align}
Interchanging the roles of $f$ and $g$ and canceling the term symmetric in $f$ and $g$, we obtain
\begin{equation}
\int\left(Lf\right)gh - \int f\left(Lg\right)h = \int\left[\left(\partial_if\right)g - \left(\partial_ig\right)f\right]\left(-{1 \over 2}\partial_j\left(\Sigma_{ij}h\right) + b_ih\right).
\end{equation}
If $h$ is a solution to the equation
\begin{equation}
-{1 \over 2}\partial_j\left(\Sigma_{ij}h\right) + b_ih = 0
\end{equation}
then the above expression is zero, showing that the operator $L$ is symmetric on the space $L^2_h$.  Conversely, for this symmetry to hold, the $\mathbb{R}^m$-valued function  $-{1 \over 2}\partial_j\left(\Sigma_{ij}h\right) + b_ih$ has to be orthogonal to all elements of the space $L^2$ (of functions with values in  $\mathbb{R}^m$) of the form $\left(\partial_if\right)g - \left(\partial_ig\right)f$.  It is not hard to prove that every $C^1$ function with this property must vanish, and thus, that $\frac{1}{2}\partial_j\left(\Sigma_{ij}h\right) - b_ih = 0$.  Here is a sketch of a proof:  suppose $\phi$ is $C^1$ and  orthogonal to all such functions.  That is, for every  $f$ and $g$, 
\begin{equation}
\int \left[\phi_i \left(\partial_i f\right)g - \phi_i\left(\partial_i g\right)f\right] = 0.
\end{equation}
Integrating the first term by parts we obtain
\begin{equation}
\int\left[-\left(\partial_i\phi_i\right)g - 2\phi_i\partial_i g\right]f = 0.
\end{equation}
Since this holds for all $f$, it follows that 
\begin{equation}
-\left(\partial_i\phi_i\right)g - 2\phi_i\partial_i g = 0
\end{equation}
and thus also 
\begin{equation}
\int\left[-\left(\partial_i\phi_i\right)g - 2\phi_i\partial_i g\right] = 0.
\end{equation}
Integrating the second term by parts, we get
\begin{equation}
\int\left(\partial_i\phi_i\right)g = 0
\end{equation}
and, since this is true for every $g$, it follows that $\partial_i\phi_i$ vanishes.  We thus have, for every $g$
\begin{equation}
\phi_i \partial_ig = 0
\end{equation}
and this implies that $\phi$ vanishes. In summary:

{\bf Proposition:}  If the density $h$ of the stationary probability measure is $C^2$, then $h$ satisfies the stationary Fokker-Planck equation
\begin{equation}
\partial_i\left[-{1 \over 2}\partial_j\left(\Sigma_{ij}h\right) + b_ih\right] = 0.
\end{equation}
The stronger statement
\begin{equation}
-{1 \over 2}\partial_j\left(\Sigma_{ij}h\right) + b_ih = 0
\end{equation}
is equivalent to symmetry of the operator $L$ on the space $L^2_h$.

We are now going to relate the above symmetry statement to the detailed balance property of the stationary dynamics.
First, it is clearly equivalent to the analogous property for the backward Kolmogorov semigroup:
\begin{equation}
\int\left(P_tf\right)gh = \int f\left(P_tg\right)h
\end{equation}
since $P_t = \exp\left(tL\right)$.
Now, $\left(P_tf\right)(x)$ is the expected value of $f(x_t)$ for the process, starting at $x$ at time 0.  In particular, for $f = \delta_y$, we obtain $P_tf(x) = p_t(x,y)$---the density of the transition probability from $x$ to $y$ in time $t$.  Using the above symmetry of $P_t$ with $f = \delta_y$ and $g = \delta_x$, we obtain the detailed balance condition:
\begin{equation}
h(x) p_t(x,y) = h(y) p_t(y, x)
\end{equation}
which, conversely, implies the symmetry statement for arbitrary $f$ and $g$.
 \subsection{The case of a linear drift and constant noise}

When both $b(y)$ and $\sigma(y)$ are constant or depend linearly on $y$,  \req{SDE} can be solved explicitly \cite{arnold} and an explicit formula for its stationary distribution can be found, when it exists.  We consider the special case $b(y) = - \gamma y$ and $\sigma(y) \equiv \sigma$, where $\gamma$ and $\sigma$ are constant matrices and  the eigenvalues of $\gamma$ have positive real parts. The stationary Fokker-Planck equation, \req{stationaryFP_explicit}, reads
\begin{equation}
\nabla \cdot \left({1 \over 2}\Sigma\nabla h + (\gamma y)h\right) = 0
\end{equation}
where $\Sigma = \sigma\sigma^T$.
It has a Gaussian solution
\begin{equation}
h(y) = \left(2\pi\right)^{-{m \over 2}}\left(\det M\right)^{-{1 \over 2}}\exp\left(-{1 \over 2}\left(M^{-1}y, y\right)\right)
\end{equation}
with the covariance matrix $M$ which is the unique solution of the Lyapunov equation
\begin{equation}
\gamma M + M\gamma^T = \Sigma
\end{equation}
and can be written as (see, for example, \cite{ortega2013matrix})
\begin{equation}
M = \int_0^{\infty}\exp\left(-t\gamma\right)\Sigma\exp\left(-t\gamma^T\right)dt.
\end{equation}
This result can be verified by a direct calculation.  We emphasize that it holds without assuming the detailed balance condition.  The latter condition is satisfied if and only if
the above $h$ solves the equation
\begin{equation}
{1 \over 2}\Sigma\nabla h + (\gamma y)h = 0
\end{equation}
which is equivalent to $M = {1 \over 2}\gamma^{-1}\Sigma$ or, in terms of the coefficients of the system, to 
\begin{align}\label{condition_forDB}
\Sigma\gamma^T = \gamma \Sigma
\end{align}
To see the physical significance of this condition, let us go back to the general case and write (adapting the discussion in  \cite{Zwanzig} to our notation)
\begin{equation}
\gamma = {1 \over 2}\Sigma M^{-1} - i\Omega.
\end{equation}
$\Omega$ represents the ``oscillatory degrees of freedom'' of the diffusive system.  The above calculations show that the detailed balance condition is equivalent to $\Omega = 0$, in agreement with the physical intuition that there are no macroscopic currents in the stationary state.
\section{Small mass limit---a perturbative approach}\label{sec:perturb}

We are now going to apply the general facts about Langevin equations to a model of a mechanical system, interacting with a noisy environment.  The dynamical variables of this system are positions and momenta, and, in general, the Langevin equations which describe its time evolution, are not linear.  However, when investigating the small mass limit of the system by a perturbative method, we will encounter equations closely related to those studied above.  This will be explained later, when we interpret the limiting equations. 

Consider a mechanical system with the Hamiltonian $\mathcal{H}(q, p)$ where $q, p \in \mathbb{R}^n$.  We want to study a small mass limit of this system, coupled to a damping force and the noise.  Therefore, we introduce the variable $z = {p \over \sqrt{m}}$ and assume the Hamiltonian can be written  $\mathcal{H}(q,p) =H(q,z)$ where $H$ is independent of $m$. We thus have
\begin{align}\label{Langevin}
dq_t &= {1 \over \sqrt{m}}\nabla_zH(q_t, z_t)\,dt \\
dz_t &= -{1 \over \sqrt{m}}\nabla_qH(q_t, z_t)\,dt - {1 \over m}\gamma(q_t)\nabla_zH(q_t, z_t)\,dt + {1 \over \sqrt{m}}\sigma(q_t)\,dW_t .\notag
\end{align}
$\gamma$ is $n \times n$-matrix-valued, $\sigma$ is $n \times k$-matrix-valued and $W$ is a $k$-dimensional Wiener process.  We emphasize that $\sigma$ does not play here the same role that it played in our discussion of the general Langevin equation, since the noise term enters only the equation for $dz_t$.  The number $k$ of the components of the driving noise does not have to be related to the dimension of the system in any particular way.  The corresponding backward Kolmogorov equation for a function $\rho(q, z, t)$ is
\begin{equation}
\partial_t \rho = L\rho
\end{equation}
where the differential operator $L$ equals
\begin{equation}
L = {1 \over m}L_1 + {1 \over \sqrt{m}}L_2
\end{equation}
with 
\begin{align}
L_1 &= {1 \over 2}\Sigma\nabla_z \cdot \nabla_z - \gamma \nabla_zH \nabla_z \\
L_2 &= \nabla_zH \cdot \nabla_q - \nabla_qH \cdot \nabla_z \notag
\end{align}
where $\Sigma(q) = \sigma(q)\sigma(q)^T$.
We represent the solution of the Kolmogorov equation as a formal series 
\begin{equation}
\rho = \rho_0 + \sqrt{m}\rho_1 + m\rho_2 + \dots
\end{equation}
Equating the expressions, proportional to $m^{-1}$, $m^{-{1 \over 2}}$ and $m^0$, we obtain the equations:
\begin{align}
L_1 \rho_0 &= 0, \\
L_1 \rho_1 &= -L_2\rho_0, \notag\\
\partial_t \rho_0 &= L_1\rho_2 + L_2\rho_1. \notag
\end{align}
To satisfy the first equation it is sufficient to choose $\rho_0$ which does not depend on $z$:
\begin{equation}
\rho_0 = \rho_0(q,t).
\end{equation}
If we now search for $\rho_1$ which is linear in $z$, the second equation simplifies to
\begin{equation}
\gamma \nabla_z H \cdot \nabla_z \rho_1 = \nabla_z H \cdot \nabla_q \rho_0
\end{equation}
which has a solution 
\begin{equation}
\rho_1(q,z) = \left(\gamma^{-1}\right)^T\nabla_q \rho_0\cdot z = \nabla_q\rho_0 \cdot \gamma^{-1}z.
\end{equation}
Writing the third equation as
\begin{equation}
\partial_t \rho_0 - L_2\rho_1 = L_1\rho_2
\end{equation}
and applying the identity 
\begin{equation}
Ran L_1 = \left(Ker L_1^*\right)^{\perp}
\end{equation}
to the space $L^2$ with respect to the $z$ variable,
we see that $\partial_t \rho_0 - L_2\rho_1 = L_1\rho_2$ must be orthogonal in this space to any function $h$ in the null space of $L_1^*$.  We have
\begin{align}\label{nullspace}
L_1^*h = \nabla_z\cdot\left({1 \over 2}\Sigma\nabla_zh + \left(\gamma\nabla_zH\right)h\right)
\end{align}
where $\Sigma = \sigma\sigma^T$.

It is impossible to continue the analysis without further, simplifying assumptions.  We are first going to study the case of a general $H$, assuming a form of the detailed balance condition in the variable $z$, at fixed $q$.  

{\bf Assumption 1:} for every $q$ there exists a nonnegative solution of the equation
\begin{align}\label{conditionalDB}
{1 \over 2}\Sigma\nabla_zh + \left(\gamma\nabla_zH\right)h = 0 
\end{align}
of finite  $L^1(dz)$-norm.  We can thus choose
\begin{equation}
\int h(q,z)\,dz = 1.
\end{equation}
We will say in this case that the system satisfies the {\it conditional detailed balance property} in the variable $z$. 
Since $\rho_0$ does not depend on $z$, the orthogonality condition can be written as
\begin{equation}
\partial_t\rho_0 = \int L_2\rho_1(q,z)h(q,z)\,dz.
\end{equation}
We have the following explicit formula for $L_2\rho_1$ (summation over repeated indices is implied):
\begin{align}
L_2\rho_1=& \partial_{z_i}H\left(\partial_{q_i}\partial_{q_j}\rho_0\right)\left(\gamma^{-1}\right)_{jk}z_k + \partial_{z_i}H\left(\partial_{q_j}\rho_0\right)\partial_{q_i}\left(\left(\gamma^{-1}\right)_{jk}\right)z_k\\
& - \partial_{q_i}H\left(\partial_{q_j}\rho_0\right)\left(\gamma^{-1}\right)_{ji}. \notag
\end{align}
To integrate it against $h(q,z)$, we will use the following consequence of \req{conditionalDB}
\begin{align}\label{averaging}
&\int\left(\partial_{z_i}H\right)z_kh(q,z)\,dz = -{1 \over 2}\int\left(\gamma^{-1}\Sigma\nabla_zh\right)_iz_k\,dz\\
 =& -\int\left(\gamma^{-1}\Sigma\right)_{ij}\left(\partial_{z_j}h\right)z_k\,dz= -{1 \over 2}\left(\gamma^{-1}\Sigma\right)_{ij}\int\left(\partial_{z_j}h\right)z_k\,dz \notag\\
=& {1 \over 2}\left(\gamma^{-1}\Sigma\right)_{ij}\int h\delta_{jk}\,dz = {1 \over 2}\left(\gamma^{-1}\Sigma\right)_{ik}.\notag
\end{align}
The orthogonality condition is thus
\begin{align}\partial_t \rho_0 =& -\left(\gamma^{-1}\right)_{ji}\left<\partial_{q_i}H\right>\partial_{q_j}\rho_0 + {1 \over 2}\left(\gamma^{-1}\Sigma\right)_{ik}\partial_{q_i}\left(\left(\gamma^{-1}\right)_{jk}\right)\partial_{q_j}\rho_0\\
&+ {1 \over 2}\left(\gamma^{-1}\Sigma\right)_{ik}\left(\gamma^{-1}\right)_{jk}\left(\partial_{q_i}\partial_{q_j}\rho_0\right).\notag
\end{align}
In this formula, which is more general than the detailed-balance case of the rigorous result of \cite{Hottovy}, $\left<-\right>$ denotes the average (i.e. the integral over $z$ with the density $h(q,z)$).  This notation is used only in the term in which the average has not been calculated explicitly.  Passing from the Kolmogorov equation to the corresponding SDE, we obtain the effective Langevin equation in the $m \to 0$ limit:
\begin{align}\label{effective}
dq_t = -\gamma(q_t)^{-1}\left(\left<\nabla_qH\right>(q_t) + S(q_t)\right)\,dt  + \gamma^{-1}(q_t)\sigma(q_t)\, dW_t
\end{align}
where the components of the noise-induced drift, $S(q)$, are given by
\begin{equation}
S_i(q) = {1 \over 2}\left(\gamma^{-1}\Sigma\right)_{jk}\partial_{q_j}\left(\left(\gamma^{-1}\right)_{ik}\right)
\end{equation}
and we have used
\begin{equation}
\gamma^{-1}\sigma\left(\gamma^{-1}\sigma\right)^T = \gamma^{-1}\Sigma\left(\gamma^{-1}\right)^T.
\end{equation}
We are now going to interpret the limiting equation \req{effective}, using the stationary probability measure $h(q,z)\,dz$, as follows:  from the original equations for $q_t$ and $z_t$ we obtain
\begin{equation}
dq_t = -\gamma(q_t)^{-1}\nabla_qH\,dt + \gamma(q_t)^{-1}\sigma(q_t)\,dW_t -\sqrt{m}\gamma(q_t)^{-1}\,dz_t.
\end{equation}
Integrating the last term by parts, we obtain
\begin{equation}
\sqrt{m}\left(\gamma^{-1}(q_t)\right)_{ij}\,dz_t^j = d\left(\sqrt{m}\left(\gamma_{ij}^{-1}(q_t)\right)z_t^j\right)-\sqrt{m}\,d\left(\left(\gamma^{-1}\right)_{ij}\right)z_t^j.
\end{equation}
We leave the first term out, since, under fairly general natural assumptions, it is of order $m^{1 \over 2}$  \cite{Hottovy}.  The second term equals
\begin{equation}
-\partial_{q_k}\left(\left(\gamma^{-1}\right)_{ij}\right)\left(\partial_{z_k}H\right)z_j\,dt.
\end{equation}
We substitute this into the equation for $dq_t$ and average, multiplying by $h(q,z)$ and integrating over $z$.  The calculation is as in \req{averaging} and the result is thus the same as the equation obtained by the multiscale expansion \req{effective}. This provides the following heuristic physical interpretation of the perturbative result:  the smaller $m$ is, the faster the variation of $z$ becomes, and in the limit $m \to 0$, $z$ homogenizes instantaneously, with  $q$ changing only infinitesimally.  

Let us now discuss conditions, under which one may expect our conditional detailed balance assumption to hold.  As seen above, at fixed $q$ this assumption is equivalent to existence of a non-negative, integrable solution of the equation 
\begin{align}\label{condDB}
{1 \over 2}\Sigma \nabla_zh + \gamma\left(\nabla_zH\right)h = 0.
\end{align}
This equation can be rewritten as
\begin{equation}
{\nabla_zh \over h} = -2\Sigma^{-1}\gamma\nabla_zH.
\end{equation}
The left-hand side equals $\nabla_z \log h$.  Letting $B = -2\Sigma^{-1}\gamma$ to simplify notation, we see that a necessary condition for existence of a solution is that $B\nabla_zH$ be a gradient.  This requires
\begin{equation}
\partial_{z_k}\left(b_{ij}\partial_{z_j}H\right) = \partial_{z_i}\left(b_{kj}\partial_{z_j}H\right)
\end{equation}
for all $i, k$, where $b_{ij}$ are matrix elements of $B$.  Introducing the matrix $R = \left(r_{ij}\right)$ of second derivatives of $H$,
\begin{equation}
r_{ij} = \partial_{z_i}\partial_{z_j}H
\end{equation}
we see that solvability of \req{condDB}  is equivalent to symmetry of the product $BR$:
\begin{equation}
BR = RB^T.
\end{equation}
For the system to satisfy the conditional detailed balance property, this relation has to be satisfied for all $q$ and $z$.  When $H$ is a quadratic function of $z$, the matrix $R$ is constant.  Even though in this case we will derive the limiting equation withouth assuming conditional detailed balance, let us remark that the above approach provides a method of determining when that condition holds, different from that used earlier.  Namely, let
\begin{equation}
H(q,z) = V(q) + {1 \over 2}Q(q)z\cdot z
\end{equation}
where $Q(q)$ is a symmetric matrix.  We then have $R = Q$ and the solvability condition becomes 
\begin{equation}
BQ = QB^T.
\end{equation}
In a still more special---but the most fequently considered---case when $Q$ is a multiple of identity, this reduces to
\begin{equation}
B = B^T
\end{equation}
which is easily seen to be equivalent to the relation
\begin{equation}
\gamma \Sigma = \Sigma\gamma^T.
\end{equation}
We have derived this condition earlier by a different argument \req{condition_forDB}.

If $\gamma$ is symmetric, this becomes the commutation relation
\begin{equation}
\gamma \Sigma = \Sigma\gamma.
\end{equation}
Note that if $\gamma \Sigma = \Sigma\gamma^T$, the solution of the Lyapunov equation 
\begin{equation}
J\gamma^T + \gamma J = \Sigma
\end{equation}
is given by $J = {1 \over 2}\gamma^{-1}\Sigma$.  In this the case the linear Langevin equation in the $z$ variable, whose conditional equilibrium at fixed value of $q$ we are studying, has no ``oscillatory degrees of freedom'', as discussed earlier (see also \cite{Zwanzig}).

In the case when $H$ is not a quadratic function of $z$, the matrix $BR(q,z)$ has to be symmetric for all $q$ and $z$, which means satisfying a continuum of conditions for every fixed $q$.    It is interesting to ask whether there exist physically natural examples in which this happens, without each $B(q)$ being a multiple of identity.  We are not going to pursue this question here.  

In the case when $B(q)$ is a multiple of identity, we can write 
\begin{equation}
A = 2\beta(q)^{-1}\gamma
\end{equation}
with $\beta(q)^{-1} = k_BT(q)$ and call the scalar function $T(q)$ {\it generalized temperature}.  The limiting Kolmogorov equation reads then
\begin{equation}
\partial_t \rho_0 = -\left(\gamma^{-1}\right)_{ji}\left<\partial_{q_i}H\right>\partial_{q_j}\rho_0 +k_BT\partial_{q_k}\left(\gamma^{-1}\right)_{jk}\partial_{q_j}\rho_0+ k_BT\left(\gamma^{-1}\right)_{ij}\left(\partial_{q_i}\partial_{q_j}\rho_0\right)
\end{equation}
and the components of the noise-induced drift are thus
\begin{equation}
S_j(q) = k_BT\partial_{q_k}\left(\gamma^{-1}\right)_{jk}.
\end{equation}
The above  applies in particular in the one-dimensional case, in which $\sigma(q)^2$ and $\gamma(q)$ are scalars and hence one is always an ($q$-dependent) multiple of the other:
\begin{align}
k_BT(q) = {\sigma(q)^2 \over 2\gamma(q)}.
\end{align}
The limiting Langevin equation is in this case
\begin{equation}
dq_t = - {\left<\nabla_qH\right> \over \gamma}\,dt - {1 \over 2}{\nabla_q \gamma \over \gamma^3}\sigma^2\,dt + {\sigma \over \gamma}\,dW_t.
\end{equation}
For a Hamiltonian equal to a sum of potential and quadratic kinetic energy, $H = V(q) + {z^2 \over 2}$, the first term equals $F \over \gamma$, where $F = -\nabla_qV\,dt$, in agreement with earlier results. 

The second situation, in which the perturbative treatment of the original system can be carried out explicitly is the general quadratic kinetic energy case.
\medskip

{\bf Assumption 2}:  $H = V(q) + {z^2 \over 2}$

If we follow the singular perturbation method used above, we again need to find the integral \req{averaging}, where $\partial_{z_i}H = z_i$.  In this case we know the solution of $L_1^*h = 0$ explicitly:
\begin{equation}
h(q,z) = \left(2\pi\right)^{-{n \over 2}}\left(\det M\right)^{-{1 \over 2}}\exp\left(-{1 \over 2}M^{-1}z\cdot z\right)
\end{equation}
so the integral in \req{averaging} is the mean of $z_iz_k$ in the Gaussian distribution with the covariance $M = \left(m_{ik}\right)$, that is, $m_{ik}$.  The second-order term in the Kolmogorov equation is thus $m_{ik}\left(\gamma^{-1}\right)_{jk}\partial_{q_i}\partial_{q_j}\rho_0$.  The corresponding Langevin equation, which has been derived rigorously in \cite{Hottovy} is in this case
\begin{align}
dq_t = -\gamma^{-1}(q_t)\nabla_q V(q_t)\,dt + S(q_t)\,dt + \gamma^{-1}(q_t)\sigma(q_t)\,dW_t.
\end{align}
The homogenization heuristics proposed under Assumption 1 applies here as well:  the limiting Langevin equation can be interpreted as a result of averaging over the conditional stationary distribution of the $z$ variable.  A rigorous result, corroborating this picture has recently been proven in \cite{clt}.

\section{A rigorous homogenization theorem}

We now develop a framework for the  homogenization of Langevin equations that is able to make many of the heuristic results from the previous two sections rigorous. Our results will concern Hamiltonians of the form
\begin{align}\label{H_form}
H(t,x)=K(t,q,p-\psi(t,q))+V(t,q)
\end{align}
where $x=(q,p)\in\mathbb{R}^n\times \mathbb{R}^n$, $K = K(t,q,z)$ and $V= V(t,q)$ are $C^2$, $\mathbb{R}$-valued functions, $K$ is non-negative, and $\psi$ is a $C^2$, $\mathbb{R}^n$-valued function.  The splitting of $H$ into $K$ and $V$ does not have to correspond physically to any notion of kinetic and potential energy, although we will use those terms for convenience.  The splitting is not unique; it will be constrained further as we continue.
We now define the family of scaled Hamiltonians, parameterized by $\epsilon>0$ (generalizing the above mass parameter):
\begin{align}
H^\epsilon(t,q,p)\equiv K^\epsilon(t,q,p)+V(t,q)\equiv K(t,q,(p-\psi(t,q))/\sqrt{\epsilon})+V(t,q).
\end{align}

Consider the following family of SDEs:\\
\begin{align}
dq^\epsilon_t=&\nabla_p H^\epsilon(t,x^\epsilon_t)dt,\label{Hamiltonian_SDE_q}\\
d p^\epsilon_t=&(-\gamma(t,x^\epsilon_t)\nabla_p H^\epsilon(t,x^\epsilon_t)-\nabla_q H^\epsilon(t,x^\epsilon_t)+F(t,x^\epsilon_t))dt+\sigma(t,x^\epsilon_t) dW_t,\label{Hamiltonian_SDE_p}
\end{align}
where $\gamma:[0,\infty)\times\mathbb{R}^{2n}\rightarrow\mathbb{R}^{n\times n}$ and $\sigma:[0,\infty)\times\mathbb{R}^{2n}\rightarrow\mathbb{R}^{n\times k}$ are continuous, $\gamma$ is positive definite, and $W_t$ is a $\mathbb{R}^k$-valued Brownian motion on a filtered probability space $(\Omega,\mathcal{F},\mathcal{F}_t,P)$ satisfying the usual conditions \cite{karatzas2014brownian}.   

Our objective in this section is to develop a method for investigating the behavior of $x^\epsilon_t$ in the limit $\epsilon\rightarrow 0^+$; more precisely, we wish to  prove the existence of a limiting ``position" process $q_t$ and derive a homogenized SDE that it satisfies.  In fact, the method we develop is  applicable to a more general class SDEs  that share certain properties with \req{Hamiltonian_SDE_q}-\req{Hamiltonian_SDE_p}.  In the following subsection, we discuss some prior results concerning \req{Hamiltonian_SDE_q}-\req{Hamiltonian_SDE_p}.  This will help motivate the assumptions made in the development of our general homogenization method, starting in subsection \label{sec:gen_homog}.

\subsection{Summary of prior results}
Let $x_t^\epsilon$ be a family of solutions to the SDE \req{Hamiltonian_SDE_q}-\req{Hamiltonian_SDE_p} with  initial condition $x_0^\epsilon=(q_0^\epsilon,p_0^\epsilon)$. We assume that a solution exists for all $t\geq 0$ (i.e. there are no explosions).  See  Appendix B in \cite{BirrellHomogenization} for assumptions that guarantee this. Under Assumptions  1-3 in \cite{BirrellHomogenization} (repeated as Assumptions \ref{assump1}-\ref{assump3} in  \ref{app:assump}, we showed that for any  $T>0$, $p>0$, $0<\beta<p/2$ we have
\begin{align}\label{results_summary}
\sup_{t\in[0,T]}E\left[\|p_t^\epsilon-\psi(t,q^\epsilon_t)\|^p\right]=O(\epsilon^{p/2}) \text{ and } E\left[\sup_{t\in[0,T]}\|p_t^\epsilon-\psi(t,q^\epsilon_t)\|^p\right]=O(\epsilon^{\beta})
\end{align}
as $\epsilon\rightarrow 0^+$ i.e. the point $(p,q)$ is attracted to the surface defined by $p=\psi(t,q)$.

Adding Assumption 4 (Assumption \req{assump4} in the appendix) we also showed that
\begin{align}\label{q_eq}
d(q_t^\epsilon)^i=&(\tilde\gamma^{-1})^{ij}(t,q_t^\epsilon)(-\partial_t\psi_j(t,q_t^\epsilon)-\partial_{q^j}V(t,q_t^\epsilon)+F_j(t,x^\epsilon_t))dt\\
&+(\tilde\gamma^{-1})^{ij}(t,q_t^\epsilon)\sigma_{j\rho}(t,x_t^\epsilon)dW^\rho_t-(\tilde\gamma^{-1})^{ij}(t,q_t^\epsilon)\partial_{q^j}K(t,q_t^\epsilon,z_t^\epsilon)dt\notag\\
&+(z_t^\epsilon)_j\partial_{q^l}(\tilde\gamma^{-1})^{ij}(t,q_t^\epsilon)\partial_{z_l}K(t,q_t^\epsilon,z_t^\epsilon)dt- d((\tilde\gamma^{-1})^{ij}(t,q_t^\epsilon)(u^\epsilon_t)_j)\notag\\
&+(u_t^\epsilon)_j\partial_t(\tilde\gamma^{-1})^{ij}(t,q_t^\epsilon)dt,\notag
\end{align}
where $u_t^\epsilon\equiv p_t^\epsilon-\psi(t,q^\epsilon_t)$, $z_t^\epsilon\equiv u_t^\epsilon/\sqrt{\epsilon}$,  and
 \begin{align}\label{tilde_gamma_def}
\tilde\gamma_{ik}(t,q)\equiv\gamma_{ik}(t,q) +\partial_{q^k}\psi_i(t,q)-\partial_{q^i}\psi_k(t,q).
\end{align}
We define the components of $\tilde\gamma^{-1}$ such that 
\begin{align}\label{tilde_gamma_inv_def}
(\tilde\gamma^{-1})^{ij}\tilde\gamma_{jk}=\delta^i_k,
\end{align}
\begin{comment}
i.e. to think of them as linear maps, we lower the second index on $\tilde\gamma^{-1}$ and raise the first index on $\tilde\gamma$
\end{comment}
and for any $v_i$ we define $(\tilde\gamma^{-1}v)^i=(\tilde\gamma^{-1})^{ij}v_j$.

Under the additional Assumptions 5-7 in \cite{BirrellHomogenization}, which include further restrictions on the form of the Hamiltonian, we were then able to show that $q_t^\epsilon$ converges in an $L^p$-norm as $\epsilon\rightarrow 0^+$ to the solution of a lower dimensional SDE,
\begin{align}\label{limit_eq}
dq_t=&\tilde \gamma^{-1}(t,q_t)(-\partial_t\psi(t,q_t)-\nabla_{q}V(t,q_t)+F(t,q_t,\psi(t,q_t)))dt+S(t,q_t)dt\notag\\
&+\tilde\gamma^{-1}(t,q_t)\sigma(t,q_t,\psi(t,q_t)) dW_t.
\end{align}
The   {\em noise-induced drift}  term, $S(t,q)$,  that arises in the limit is the term of greatest interest here. Its form is given in Eq. (3.26) in \cite{BirrellHomogenization}.  

The homogenization technique used in \cite{BirrellHomogenization} to arrive at \req{limit_eq} relies heavily on the specific structural assumptions on the form of the Hamiltonian.  Those assumptions cover a wide variety of important systems, such as a particle in an electromagnetic field, and motion on a Riemannian manifold, but it is desirable to search for a more generally applicable homogenization method.  In this paper, we develop a significantly more general  technique, adapted from the methods presented in \cite{pavliotis2008multiscale}, that is capable of homogenizing  terms of the form $G(t,q_t^\epsilon,(p_t^\epsilon-\psi(t,q_t^\epsilon))/\sqrt{\epsilon})dt$
for a general class of SDEs that satisfy the property \req{results_summary}, as well as prove convergence of $q_t^\epsilon$ to the solution of a limiting, homogenized SDE. In particular, it will be capable of homogenizing $q_t^\epsilon$ from the Hamiltonian system \req{Hamiltonian_SDE_q}-\req{Hamiltonian_SDE_p} under less restrictive assumptions on the form of the Hamiltonian, than those made in \cite{BirrellHomogenization}.  We emphasize that the convergence statements are proven in the strong sense, see Section \ref{sec:gen_homog}.

\subsection{General homogenization framework}\label{sec:gen_homog}
Here we describe our homogenization technique in a more general context  than the Hamiltonian setting from the previous section. This method is related to the cell problem method from \cite{pavliotis2008multiscale}, our proof applies to a larger class of SDEs and demonstrates $L^p$-convergence rather that weak convergence.

We will denote an element of $\mathbb{R}^{n}\times\mathbb{R}^m$ by $x=(q,p)$, where we no longer require the $q$ and $p$ degrees of freedom to have the same dimensionality, though we still employ the convention of writing $q$ indices with superscripts and $p$ indices with subscripts.  We let $W_t$ be an $\mathbb{R}^k$-valued Wiener process, $\psi:[0,\infty)\times\mathbb{R}^n\rightarrow\mathbb{R}^m$ be $C^2$ and $G_1,F_1:[0,\infty)\times\mathbb{R}^{n+m}\times\mathbb{R}^m\rightarrow\mathbb{R}^n$, $G_2,F_2:[0,\infty)\times\mathbb{R}^{n+m}\times\mathbb{R}^m\rightarrow\mathbb{R}^m$, $\sigma_1:[0,\infty)\times\mathbb{R}^{n+m}\rightarrow\mathbb{R}^{n\times k}$, and $\sigma_2:[0,\infty)\times\mathbb{R}^{n+m}\rightarrow\mathbb{R}^{m\times k}$ be continuous. With these definitions, we consider the following family of SDEs, depending on a parameter $\epsilon>0$:
\begin{align}
dq^\epsilon_t=&\left(\frac{1}{\sqrt{\epsilon}}G_1(t,x_t^\epsilon,z_t^\epsilon)+F_1(t,x_t^\epsilon,z_t^\epsilon)\right)dt+\sigma_1(t,x_t^\epsilon)dW_t,\label{gen_SDE1}\\
d p^\epsilon_t=&\left(\frac{1}{\sqrt{\epsilon}}G_2(t,x^\epsilon_t,z_t^\epsilon)+F_2(t,x^\epsilon_t,z_t^\epsilon)\right)dt+\sigma_2(t,x^\epsilon_t) dW_t,\label{gen_SDE2}
\end{align}
where we define $z_t^\epsilon=(p_t^\epsilon-\psi(t,q_t^\epsilon))/\sqrt{\epsilon}$. We will assume, in analogy with \req{results_summary}, that:
\begin{assumption}\label{homog_assump1}
For any  $T>0$, $p>0$, $0<\beta<p/2$ we have
\begin{align}
\sup_{t\in[0,T]}E\left[\|p_t^\epsilon-\psi(t,q^\epsilon_t)\|^p\right]=O(\epsilon^{p/2}) \text{ and } E\left[\sup_{t\in[0,T]}\|p_t^\epsilon-\psi(t,q^\epsilon_t)\|^p\right]=O(\epsilon^{\beta})
\end{align}
as $\epsilon\rightarrow 0^+$.
\end{assumption}
In  words, we assume that the $p$ degrees of freedom are attracted to the values defined by $p=\psi(t,q)$.  This is an appropriate setting to expect some form of homogenization, as it suggests that the  dynamics in the limit $\epsilon\rightarrow 0^+$ can be characterized by fewer degrees of freedom---the $q$-variables.

\subsubsection{Homogenization of integral processes}
In this section we derive a method capable of homogenizing processes of the form 
\begin{equation}\label{M_t_def}
M^\epsilon_t\equiv \int_0^tG(s,x_s^\epsilon,z_s^\epsilon)ds
\end{equation}
 in the limit $\epsilon\rightarrow 0^+$.  More specifically, our aim is to find conditions under which there exists some function, $S(t,q)$, such that
\begin{align}\label{homog_goal}
\int_0^tG(s,x_s^\epsilon,z_s^\epsilon)ds-\int_0^tS(s,q_s^\epsilon)ds\rightarrow 0
\end{align}
in some norm, as $\epsilon\rightarrow 0^+$ , i.e. only the $q$-degrees of freedom are needed to characterize $M_t^\epsilon$ in the limit.  We will call  a family of processes, $S(t,q_t^\epsilon)dt$, that satisfies such a limit, a homogenization of $G(t,x_t^\epsilon,z_t^\epsilon)dt$.    The technique we develop will also be useful for proving existence of a limiting process $q_s$ (i.e. $q_s^\epsilon\rightarrow q_s$), and  showing that
\begin{align}
\int_0^tG(s,x_s^\epsilon,z_s^\epsilon)ds\rightarrow \int_0^tS(s,q_s)ds.
\end{align}
as $\epsilon\rightarrow 0^+$. We will consider this second question in Section \ref{sec:limit_eq}. Here, our focus is on  \req{homog_goal}.

 As a starting point, let $\chi(t,x,z):[0,\infty)\times\mathbb{R}^{n+m}\times\mathbb{R}^m\rightarrow\mathbb{R}$ be $C^{1,2}$, where $C^{1,2}$ is defined as follows:
\begin{itemize}
\item If $\sigma_1\neq 0$ then we take this to mean $\chi$ is $C^1$ and, for each $t$, $\chi(t,x,z)$ is $C^2$ in $(x,z)$  with second derivatives continuous jointly in all variables.
\item  If $\sigma_1=0$ then  we take this to mean $\chi$ is $C^1$ and, for each $t,q$,  $\chi(t,q,p,z)$ is $C^2$ in $(p,z)$ with second derivatives continuous jointly in all variables.
\end{itemize}

  Eventually, we will need to carefully choose $\chi$ so that we  achieve our aim, but for now we simply use It\^o's  formula to compute $\chi(t,x_t^\epsilon,z_t^\epsilon)$. We defined $C^{1,2}$ precisely so that It\^o's formula is justified. For this computation, we will define $\chi^\epsilon(t,x)=\chi(t,x,(p-\psi(t,q))/\sqrt{\epsilon})$, and
\begin{align}\label{sigma_defs}
&\Sigma_{11}^{ij}= \sum_\rho (\sigma_1)^i_\rho(\sigma_1)^j_\rho,\hspace{1mm}  (\Sigma_{12})^i_j= \sum_\rho (\sigma_1)^i_\rho(\sigma_2)_{j\rho}, \hspace{1mm}(\Sigma_{22})_{ij}= \sum_\rho (\sigma_2)_{i\rho}(\sigma_2)_{j\rho}.
\end{align}
It\^o's  formula gives
\begin{align}
&\chi(t,x_t^\epsilon,z_t^\epsilon)=\chi(0,x_0^\epsilon,z_0^\epsilon)+\int_0^t\partial_s \chi^\epsilon(s,x_s^\epsilon)ds+\int_0^t\nabla_q\chi^\epsilon(s,x_s^\epsilon)\cdot dq_s^\epsilon\\
&+\int_0^t\nabla_p\chi^\epsilon(s,x_s^\epsilon)\cdot dp_s^\epsilon+\frac{1}{2}\int_0^t\partial_{q^i}\partial_{q^j}\chi^\epsilon(s,x_s^\epsilon) \Sigma_{11}^{ij}(s,x_s^\epsilon)ds\notag\\
&+\frac{1}{2}\int_0^t\partial_{q^i}\partial_{p_j}\chi^\epsilon(s,x_s^\epsilon) (\Sigma_{12})^i_j(s,x_s^\epsilon)ds+\frac{1}{2}\int_0^t\partial_{p_i}\partial_{q^j}\chi^\epsilon(s,x_s^\epsilon) (\Sigma_{12})_i^j(s,x_s^\epsilon)ds\notag\\
&+\frac{1}{2}\int_0^t\partial_{p_i}\partial_{p_j}\chi^\epsilon(s,x_s^\epsilon) (\Sigma_{22})_{ij}(s,x_s^\epsilon)ds.\notag
\end{align}
Note that if $\sigma_1=0$ then only the second derivatives that we have assumed exist are involved in this computation.

We can compute these terms as follows:
\begin{align}
\partial_t \chi^\epsilon(t,x)=&\partial_t\chi(t,x,z)-\partial_{z_i}\chi(t,x,z) \partial_t\psi_i(t,q)/\sqrt{\epsilon},\\
\partial_{q^i}\chi^\epsilon(t,x)=&(\partial_{q^i}\chi)(t,x,z)-\epsilon^{-1/2}\partial_{q^i}\psi_k(t,q)(\partial_{z_k}\chi)(t,x,z),\\
\partial_{p_i}\chi^\epsilon(t,x)=&(\partial_{p_i}\chi)(t,x,z)+\epsilon^{-1/2}(\partial_{z_i}\chi)(t,x,z),\\
\partial_{q^i}\partial_{q^j}\chi^\epsilon(t,x)=&(\partial_{q^i}\partial_{q^j}\chi)(t,x,z)+\epsilon^{-1/2}\left(-\partial_{q^j}\psi_k(t,q)(\partial_{q^i}\partial_{z_k}\chi)(t,x,z)\right.\\
&\left.-\partial_{q^i}\partial_{q^j}\psi_k(t,q)(\partial_{z_k}\chi)(t,x,z)-\partial_{q^i}\psi_k(t,q)(\partial_{q^j}\partial_{z_k}\chi)(t,x,z)\right)\notag\\
&+\epsilon^{-1}\partial_{q^i}\psi_k(t,q)\partial_{q^j}\psi_l(t,q)(\partial_{z_k}\partial_{z_l}\chi)(t,x,z).\notag\\
\partial_{p_i}\partial_{p_j}\chi^\epsilon(t,x)=&(\partial_{p_i}\partial_{p_j}\chi)(t,x,z)+\epsilon^{-1/2}\left((\partial_{z_j}\partial_{p_i}\chi)(t,x,z)\right.\\
&\left.+(\partial_{p_j}\partial_{z_i}\chi)(t,x,z)\right)+\epsilon^{-1}(\partial_{z_i}\partial_{z_j}\chi)(t,x,z),\notag\\
\partial_{q^i}\partial_{p_j}\chi^\epsilon(t,x)=&(\partial_{q^i}\partial_{p_j}\chi)(t,x,z)+\epsilon^{-1/2}\left((\partial_{q^i}\partial_{z_j}\chi)(t,x,z)\right.\\
&\left.-\partial_{q^i}\psi_k(t,q)(\partial_{p_j}\partial_{z_k}\chi)(t,x,z)\right)\notag\\
&-\epsilon^{-1}\partial_{q^i}\psi_k(t,q)(\partial_{z_j}\partial_{z_k}\chi)(t,x,z),\notag
\end{align}
where $z$ is evaluated at $z(t,x,\epsilon)=(p-\psi(t,q))/\sqrt{\epsilon}$ in each of the above formulae.

Using these expressions, together with the SDE \req{gen_SDE1}-\req{gen_SDE2} we find
\begin{align}
&\chi(t,x_t^\epsilon,z_t^\epsilon)=\chi(0,x_0^\epsilon,z_0^\epsilon)+\int_0^t\partial_s \chi(s,x_s^\epsilon,z_s^\epsilon)-\epsilon^{-1/2} (\partial_{z_i}\chi)(s,x_s^\epsilon,z_s^\epsilon) \partial_s\psi_i(s,q_s^\epsilon)ds\\
&+\int_0^t\left((\partial_{q^i}\chi)(s,x_s^\epsilon,z_s^\epsilon)-\epsilon^{-1/2}\partial_{q^i}\psi_k(s,q_s^\epsilon)(\partial_{z_k}\chi)(s,x_s^\epsilon,z_s^\epsilon)\right) \notag\\
&\hspace{10mm}\times\left[\left(\frac{1}{\sqrt{\epsilon}}G_1(s,x_s^\epsilon,z_s^\epsilon)+F_1(s,x_s^\epsilon,z_s^\epsilon)\right)ds+\sigma_1(s,x_s^\epsilon)dW_s\right]^i\notag\\
&+\int_0^t\left((\partial_{p_i}\chi)(s,x_s^\epsilon,z_s^\epsilon)+\epsilon^{-1/2}(\partial_{z_i}\chi)(s,x_s^\epsilon,z_s^\epsilon)\right) \notag\\
&\hspace{10mm}\times\left [\left(\frac{1}{\sqrt{\epsilon}}G_2(s,x^\epsilon_s,z_s^\epsilon)+F_2(s,x^\epsilon_s,z_s^\epsilon)\right)ds+\sigma_2(s,x^\epsilon_s) dW_s\right]_i\notag\\
&+\frac{1}{2}\int_0^t\Sigma_{11}^{ij}(s,x_s^\epsilon)\bigg[(\partial_{q^i}\partial_{q^j}\chi)(s,x_s^\epsilon,z_s^\epsilon)+\epsilon^{-1/2}\left(-\partial_{q^j}\psi_k(s,q_s^\epsilon)(\partial_{q^i}\partial_{z_k}\chi)(s,x_s^\epsilon,z_s^\epsilon)\right.\notag\\
&\hspace{15mm}\left.-\partial_{q^i}\partial_{q^j}\psi_k(s,q_s^\epsilon)(\partial_{z_k}\chi)(s,x_s^\epsilon,z_s^\epsilon)-\partial_{q^i}\psi_k(s,q_s^\epsilon)(\partial_{q^j}\partial_{z_k}\chi)(s,x_s^\epsilon,z_s^\epsilon)\right)\notag\\
&\hspace{15mm}+\epsilon^{-1}\partial_{q^i}\psi_k(s,q_s^\epsilon)\partial_{q^j}\psi_l(s,q_s^\epsilon)(\partial_{z_k}\partial_{z_l}\chi)(s,x_s^\epsilon,z_s^\epsilon)\bigg] ds\notag\\
&+\int_0^t(\Sigma_{12})^i_j(s,x_s^\epsilon)\bigg[(\partial_{q^i}\partial_{p_j}\chi)(s,x_s^\epsilon,z_s^\epsilon)+\epsilon^{-1/2}\left((\partial_{q^i}\partial_{z_j}\chi)(s,x_s^\epsilon,z_s^\epsilon)\right.\notag\\
&\hspace{10mm}\left.-\partial_{q^i}\psi_k(t,q)(\partial_{p_j}\partial_{z_k}\chi)(s,x_s^\epsilon,z_s^\epsilon)\right)-\epsilon^{-1}\partial_{q^i}\psi_k(s,q_s^\epsilon)(\partial_{z_j}\partial_{z_k}\chi)(s,x_s^\epsilon,z_s^\epsilon)\bigg]ds\notag\\
&+\frac{1}{2}\int_0^t (\Sigma_{22})_{ij}(s,x_s^\epsilon)\bigg[(\partial_{p_i}\partial_{p_j}\chi)(s,x_s^\epsilon,z_s^\epsilon)+\epsilon^{-1/2}\left((\partial_{z_j}\partial_{p_i}\chi)(s,x_s^\epsilon,z_s^\epsilon)\right.\notag\\
&\hspace{15mm}\left.+(\partial_{p_j}\partial_{z_i}\chi)(s,x_s^\epsilon,z_s^\epsilon)\right)+\epsilon^{-1}(\partial_{z_i}\partial_{z_j}\chi)(s,x_s^\epsilon,z_s^\epsilon)\bigg]ds.\notag
\end{align}
Multiplying by $\epsilon$ and collecting powers, we arrive at
\begin{align}\label{homog_eq}
&\int_0^t(L\chi)(s,x_s^\epsilon,z_s^\epsilon)ds=\epsilon^{1/2} (R_1^\epsilon)_t+\epsilon\left(\chi(t,x_t^\epsilon,z_t^\epsilon)-\chi(0,x_0^\epsilon,z_0^\epsilon)+ (R^\epsilon_2)_t\right),
\end{align}
where we define
\begin{align}\label{L_def}
(L\chi)(t,x,z)=&\bigg(\frac{1}{2}\Sigma_{11}^{ij}(t,x)\partial_{q^i}\psi_k(t,q)\partial_{q^j}\psi_l(t,q)\\
&\hspace{5mm}-(\Sigma_{12})^i_l(t,x)\partial_{q^i}\psi_k(t,q)+\frac{1}{2} (\Sigma_{22})_{kl}(t,x)\bigg)(\partial_{z_k}\partial_{z_l}\chi)(t,x,z)\notag\\
&+\left((G_2)_k(t,x,z)-\partial_{q^i}\psi_k(t,q)G^i_1(t,x,z)\right)(\partial_{z_k}\chi)(t,x,z),\notag
\end{align}
\begin{align}\label{R1_def}
&(R_1^\epsilon)_t\\
=&\int_0^t (\partial_{z_i}\chi)(s,x_s^\epsilon,z_s^\epsilon) \partial_s\psi_i(s,q_s^\epsilon)ds-\int_0^t(\partial_{q^i}\chi)(s,x_s^\epsilon,z_s^\epsilon)G^i_1(s,x_s^\epsilon,z_s^\epsilon)ds\notag\\
&+\int_0^t\partial_{q^i}\psi_k(s,q_s^\epsilon)(\partial_{z_k}\chi)(s,x_s^\epsilon,z_s^\epsilon)\left[F_1(s,x_s^\epsilon,z_s^\epsilon)ds+\sigma_1(s,x_s^\epsilon)dW_s\right]^i\notag\\
&-\int_0^t(\partial_{z_i}\chi)(s,x_s^\epsilon,z_s^\epsilon) \left [F_2(s,x^\epsilon_s,z_s^\epsilon)ds+\sigma_2(s,x^\epsilon_s) dW_s\right]_i\notag\\
&-\int_0^t(\partial_{p_i}\chi)(s,x_s^\epsilon,z_s^\epsilon)(G_2)_i(s,x^\epsilon_s,z_s^\epsilon)ds\notag\\
&-\frac{1}{2}\int_0^t\Sigma_{11}^{ij}(s,x_s^\epsilon)\left(-\partial_{q^j}\psi_k(s,q_s^\epsilon)(\partial_{q^i}\partial_{z_k}\chi)(s,x_s^\epsilon,z_s^\epsilon)\right.\notag\\
&\hspace{8.5mm}\left.-\partial_{q^i}\partial_{q^j}\psi_k(s,q_s^\epsilon)(\partial_{z_k}\chi)(s,x_s^\epsilon,z_s^\epsilon)-\partial_{q^i}\psi_k(s,q_s^\epsilon)(\partial_{q^j}\partial_{z_k}\chi)(s,x_s^\epsilon,z_s^\epsilon)\right)ds\notag\\
&-\int_0^t(\Sigma_{12})^i_j(s,x_s^\epsilon)\left((\partial_{q^i}\partial_{z_j}\chi)(s,x_s^\epsilon,z_s^\epsilon)-\partial_{q^i}\psi_k(t,q)(\partial_{p_j}\partial_{z_k}\chi)(s,x_s^\epsilon,z_s^\epsilon)\right)ds\notag\\
&-\int_0^t (\Sigma_{22})_{ij}(s,x_s^\epsilon)(\partial_{z_j}\partial_{p_i}\chi)(s,x_s^\epsilon,z_s^\epsilon)ds,\notag
\end{align}
and
\begin{align}\label{R2_def}
& (R^\epsilon_2)_t\\
=&-\int_0^t\partial_s \chi(s,x_s^\epsilon,z_s^\epsilon)ds\notag\\
&-\int_0^t(\partial_{q^i}\chi)(s,x_s^\epsilon,z_s^\epsilon)\left[F_1(s,x_s^\epsilon,z_s^\epsilon)ds+\sigma_1(s,x_s^\epsilon)dW_s\right]^i\notag\\
&-\int_0^t(\partial_{p_i}\chi)(s,x_s^\epsilon,z_s^\epsilon) \left [F_2(s,x^\epsilon_s,z_s^\epsilon)ds+\sigma_2(s,x^\epsilon_s) dW_s\right]_i\notag\\
&-\frac{1}{2}\int_0^t\Sigma_{11}^{ij}(s,x_s^\epsilon)(\partial_{q^i}\partial_{q^j}\chi)(s,x_s^\epsilon,z_s^\epsilon)ds\notag\\
&-\int_0^t(\Sigma_{12})^i_j(s,x_s^\epsilon)(\partial_{q^i}\partial_{p_j}\chi)(s,x_s^\epsilon,z_s^\epsilon)ds\notag\\
&-\frac{1}{2}\int_0^t (\Sigma_{22})_{ij}(s,x_s^\epsilon)(\partial_{p_i}\partial_{p_j}\chi)(s,x_s^\epsilon,z_s^\epsilon)ds.\notag
\end{align}

First, think of simply homogenizing \req{M_t_def} to a quantity of the form $\int_0^t\tilde G(s,x_s^\epsilon)ds$.  Suppose we have a candidate for $\tilde G$.  If we can find a $C^{1,2}$ solution, $\chi$, to the PDE
\begin{align}
(L\chi)(t,x,z)=G(t,x,z)-\tilde G(t,x)
\end{align}
then substituting this into \req{homog_eq} gives
\begin{align}\label{homog_eq2}
&\int_0^t G(s,x_s^\epsilon,z_s^\epsilon)ds-\int_0^t \tilde G(s,x_s^\epsilon)ds\\
=&\epsilon^{1/2}(R_1^\epsilon)_t+\epsilon\left(\chi(t,x_t^\epsilon,z_t^\epsilon)-\chi(0,x_0^\epsilon,z_0^\epsilon)+ (R^\epsilon_2)_t\right).\notag
\end{align}
Given sufficient  growth bounds for $\chi$ and its derivatives, one anticipates that the right hand side of \req{homog_eq2} vanishes in the limit.  If in 
addition, $\tilde G$ is Lipschitz in $p$, uniformly in $(t,q)$, then, based on Assumption \ref{homog_assump1},  one expects
\begin{align}
\int_0^t G(s,x_s^\epsilon,z_s^\epsilon)ds- \int_0^t\tilde G(s,q_s^\epsilon,\psi(s,q_s^\epsilon))ds\rightarrow 0
\end{align}
as $\epsilon\rightarrow 0^+$.

We make this informal discussion precise in Theorem \ref{homog_thm}, below.  For this, we will need the following assumptions:
\begin{assumption}\label{homog_assump2}
For all $T>0$, the following quantities are polynomially bounded in $z$, with the bounds uniform on $[0,T]\times \mathbb{R}^{n+m}$:\\
$G_1$, $F_1$, $G_2$, $F_2$, $\sigma_{1}$, $\sigma_{2}$, $\partial_t\psi$, $\partial_{q^i}\psi$, $\partial_{q^i}\partial_{q^j}\psi$. If $\sigma_1=0$ then we can remove the requirement on  $\partial_{q^i}\partial_{q^j}\psi$.
\end{assumption}
Recall that an $\mathbb{R}^l$-valued function, $\phi(t,x,z)$, is called {\em polynomially bounded} in $z$, uniformly on $[0,T]\times \mathbb{R}^{n+m}$ if there exists $q,C>0$ such that
\begin{equation}
\|\phi(t,x,z)\|\leq C(1+\|z\|^q)
\end{equation}
for all $(t,x,z)\in[0,T]\times \mathbb{R}^{n+m}\times\mathbb{R}^m$.  In particular, if $\phi$ is independent of $z$, this just means it is bounded on $[0,T]\times \mathbb{R}^{n+m}$. Applying this to $\psi$, we note that Assumption \ref{homog_assump2} implies  $\psi$ is Lipschitz in $q$, uniformly in $t\in[0,T]$.

\begin{assumption}\label{homog_assump3}
 Given a continuous $G:[0,\infty)\times\mathbb{R}^{n+m}\times\mathbb{R}^m\rightarrow\mathbb{R}$, assume that there exists a $C^{1,2}$ function $\chi:[0,\infty)\times\mathbb{R}^{n+m}\times\mathbb{R}^m\rightarrow\mathbb{R}$ and a continuous function $\tilde G(t,x):[0,\infty)\times\mathbb{R}^{n+m}\rightarrow\mathbb{R}$ that together satisfy the PDE 
\begin{align}\label{chi_eq}
(L\chi)(t,x,z)= G(t,x,z)-\tilde G(t,x),
\end{align}
where the differential operator, $L$, is defined in \req{L_def}. 

Assume  that, for a given $T>0$, $\tilde G$ is Lipschitz in $p$, uniformly for $(t,q)\in[0,T]\times\mathbb{R}^n$. Also suppose that  $\chi$, its first derivatives, and the second derivatives $\partial_{q^i}\partial_{q^j}\chi$, $\partial_{q^i}\partial_{p_j}\chi$, $\partial_{q^i}\partial_{z_j}\chi$, $\partial_{p_i}\partial_{p_j}\chi$, and $\partial_{p_i}\partial_{z_j}\chi$ are polynomially bounded in $z$, uniformly for $(t,x)\in[0,T]\times\mathbb{R}^{n+m}$.   If $\sigma_1=0$ then the only second derivatives that we require to be polynomially bounded are $\partial_{p_i}\partial_{p_j}\chi$ and $\partial_{p_i}\partial_{z_j}\chi$.

\end{assumption}

\begin{theorem}\label{homog_thm}
  Fix $T>0$. Let Assumptions \ref{homog_assump1}-\ref{homog_assump3}  hold and $x_t^\epsilon=(q_t^\epsilon,p_t^\epsilon)$ satisfy the SDE \req{gen_SDE1}-\req{gen_SDE2}.   Then for any  $p>0$ we have
\begin{align}
E\left[\sup_{t\in[0,T]}\left|\int_0^t G(s,x_s^\epsilon,z_s^\epsilon)ds-\int_0^t \tilde G(s,x_s^\epsilon)ds\right|^p\right]=O(\epsilon^{p/2}).
\end{align}
and
\begin{align}
E\left[\sup_{t\in[0,T]}\left|\int_0^t G(s,x_s^\epsilon,z_s^\epsilon)ds-\int_0^t \tilde G(s,q_s^\epsilon,\psi(s,q_s^\epsilon))ds\right|^p\right]=O(\epsilon^{p/2})
\end{align}
as $\epsilon\rightarrow 0^+$.
\end{theorem}
\begin{proof}
Fix $T>0$. First let $p\geq 2$. \req{homog_eq2} gives
\begin{align}
&E\left[\sup_{t\in[0,T]}\left|\int_0^t G(s,x_s^\epsilon,z_s^\epsilon)ds-\int_0^t \tilde G(s,x_s^\epsilon)ds\right|^p\right]\\
\leq &3^{p-1}\left( \epsilon^{p/2}E\left[\sup_{t\in[0,T]}|(R_1^\epsilon)_t|^p\right]+2^p\epsilon^pE\left[\sup_{t\in[0,T]}|\chi(t,x_t^\epsilon,z_t^\epsilon)|^p\right]\right.\notag\\
&\left.\hspace{1cm}+\epsilon^p E\left[\sup_{t\in[0,T]}|(R^\epsilon_2)_t|^p\right]\right).\notag
\end{align}

From \req{R1_def} and \req{R2_def} we see that $R_1^\epsilon$ and $R_2^\epsilon$ have the forms
\begin{align}
(R_i^\epsilon)_t=\int_0^t V_i(s,x_s^\epsilon,z_s^\epsilon) ds+\int_0^t Q_{ij}(s,x_s^\epsilon,z_s^\epsilon) dW^j_s,
\end{align}
where $V_i$ and $Q_{ij}$ are linear combinations of products of (components of) one or more terms from the following list:\\
 $G_1$, $F_1$, $G_2$, $F_2$, $\sigma_1$, $\sigma_2$, $\partial_t\psi$, $\partial_{q^i}\psi$, $\partial_{q^i}\partial_{q^j}\psi$, $\partial_t\chi$, $\partial_{q^i}\chi$, $\partial_{z_i}\chi$, $\partial_{p_i}\chi$, $\partial_{q^i}\partial_{q^j}\chi$, $\partial_{q^i}\partial_{p_j}\chi$, $\partial_{q^i}\partial_{z_j}\chi$, $\partial_{p_i}\partial_{p_j}\chi$, $\partial_{p_i}\partial_{z_j}\chi$. Also note that if $\sigma_1=0$ then the only second derivatives terms that are involved are $\partial_{p_i}\partial_{p_j}\chi$ and $\partial_{p_i}\partial_{z_j}\chi$.

By assumption, these are all polynomially bounded  in $z$, uniformly on $[0,T]\times \mathbb{R}^{n+m}$, as is $\chi$.   Therefore, letting $\tilde C$ denote a constant that potentially varies line to line, there exists $r>0$ such that
\begin{align}
&E\left[\sup_{t\in[0,T]}\left|\int_0^t G(s,x_s^\epsilon,z_s^\epsilon)ds-\int_0^t \tilde G(s,x_s^\epsilon)ds\right|^p\right]\\
\leq &\tilde C \epsilon^{p/2}\left(E\left[\left(\int_0^T |V_1(s,x_s^\epsilon,z_s^\epsilon)| ds\right)^p\right]+E\left[\sup_{t\in[0,T]}\left|\int_0^t Q_{1j}(s,x_s^\epsilon,z_s^\epsilon) dW^j_s\right|^p\right]\right)\notag\\
&+\tilde C\epsilon^p\left(E\left[\left(\int_0^T |V_2(s,x_s^\epsilon,z_s^\epsilon)| ds\right)^p\right]+E\left[\sup_{t\in[0,T]}\left|\int_0^t Q_{2j}(s,x_s^\epsilon,z_s^\epsilon) dW^j_s\right|^p\right]\right.\notag\\
&\left.\hspace{15mm}+1+E\left[\sup_{t\in[0,T]}\|z_t^\epsilon\|^{rp}\right]\right).\notag
\end{align}
H\"older's inequality and polynomial boundedness yields
\begin{align}
E\left[\left(\int_0^T |V_i(s,x_s^\epsilon,z_s^\epsilon)| ds\right)^p\right]\leq& T^{p-1} E\left[ \int_0^T |V_i(s,x_s^\epsilon,z_s^\epsilon)|^p ds\right]\\
\leq & \tilde C T^{p}\left( 1+\sup_{t\in[0,T]}E\left[\|z_t^\epsilon\|^{rp} \right]\right).\notag
\end{align}
Applying the Burkholder-Davis-Gundy inequality to the terms involving $Q_{ij}$, (as found in, for example, Theorem 3.28 in \cite{karatzas2014brownian}), and then H\"older's inequality, we obtain
\begin{align}
&E\left[\sup_{t\in[0,T]}\left|\int_0^t Q_{ij}(s,x_s^\epsilon,z_s^\epsilon) dW^j_s\right|^p\right]\\
\leq &\tilde CE\left[\left( \int_0^T\| Q_i(s,x_s^\epsilon,z_s^\epsilon)\|^2ds\right)^{p/2}\right]\notag\\
\leq &\tilde C T^{p/2-1}E\left[ \int_0^T\| Q_i(s,x_s^\epsilon,z_s^\epsilon)\|^pds\right]\notag\\
\leq &\tilde C T^{p/2}\left( 1+\sup_{t\in[0,T]}E[\|z_t^\epsilon\|^{rp}]\right).\notag
\end{align}
Combining these bounds, and using  Assumption \ref{homog_assump1}, we find
\begin{align}
&E\left[\sup_{t\in[0,T]}\left|\int_0^t G(s,x_s^\epsilon,z_s^\epsilon)ds-\int_0^t \tilde G(s,x_s^\epsilon)ds\right|^p\right]\\
\leq &\tilde C \epsilon^{p/2}\left(1+\sup_{t\in[0,T]}E[\|z_t^\epsilon\|^{rp}]\right)\notag\\
&+\tilde C\epsilon^p\left(1+\sup_{t\in[0,T]}E[\|z_t^\epsilon\|^{rp}]+E\left[\sup_{t\in[0,T]}\|z_t^\epsilon\|^{rp}\right]\right)\notag\\
\leq &\tilde C \epsilon^{p/2}(1+O(1))+\tilde C\epsilon^p\left(1+O(1)+O(\epsilon^{-\delta})\right)\notag
\end{align}
for any $\delta>0$.  Letting $\delta=p/2$ we find
\begin{align}
E\left[\sup_{t\in[0,T]}\left|\int_0^t G(s,x_s^\epsilon,z_s^\epsilon)ds-\int_0^t \tilde G(s,x_s^\epsilon)ds\right|^p\right]=O(\epsilon^{p/2}).
\end{align}

Now use H\"older's inequality, the uniform Lipschitz property of $\tilde G$, and Assumption \ref{homog_assump1} again to compute
\begin{align}
&E\left[\sup_{t\in[0,T]}\left|\int_0^t G(s,x_s^\epsilon,z_s^\epsilon)ds-\int_0^t \tilde G(s,q_s^\epsilon,\psi(s,q_s^\epsilon))ds\right|^p\right]\\
\leq& O(\epsilon^{p/2})+\tilde C E\left[\sup_{t\in[0,T]}\left|\int_0^t \tilde G(s,x_s^\epsilon)-\tilde G(s,q_s^\epsilon,\psi(s,q_s^\epsilon))ds\right|^p\right]\notag\\
\leq& O(\epsilon^{p/2})+ \tilde CT^{p-1}E\left[\int_0^T |\tilde G(s,x_s^\epsilon)-\tilde G(s,q_s^\epsilon,\psi(s,q_s^\epsilon))|^p ds\right]\notag\\
\leq& O(\epsilon^{p/2})+\tilde C T^{p}\sup_{t\in[0,T]}E\left[ \|p_t^\epsilon-\psi(t,q_t^\epsilon)\|^p \right]\notag\\
=&O(\epsilon^{p/2}).\notag
\end{align}
This proves the claim for $p\geq 2$.  The result for arbitrary $p>0$ then follows from an application of H\"older's inequality.
\end{proof}

\subsubsection{Formal derivation of $\tilde G$}\label{sec:formal_G_tilde}
Formally applying the  Fredholm alternative to \req{chi_eq} motivates the form that $\tilde G$ must have in order for $\chi$ and its derivatives to possess the  growth bounds required by Theorem \ref{homog_thm}.  The formal calculation is simple enough that we repeat it here:\\
  Let $L^*$ be the formal adjoint to $L$ and suppose we have a  solution, $h(t,x,z)$, to
\begin{equation}\label{h_eq}
 L^*h=0, \hspace{2mm} \int h(t,x,z)dz=1.
\end{equation}
  If $\chi$ and its derivatives grow slowly enough and $h$ and its derivatives  decay quickly enough, then $\int hL\chi dz$ will exist, the boundary terms from integration by parts will vanish at infinity, and we find
\begin{align}
0=\int (L^*h)\chi dz=\int h L(\chi) dz=\int h (G-\tilde G)dz=\int h Gdz-\tilde G.
\end{align} 
Therefore we must have
\begin{align}
\tilde G(t,x)=\int h(t,x,z)G(t,x,z)dz.
\end{align}
In essence, the homogenized quantity is obtained by averaging over $h$, the instantaneous equilibrium distribution for the fast variables, $z$.   This corroborates the heuristic discussion in Section \ref{sec:perturb}.

\subsubsection{Limiting equation}\label{sec:limit_eq}
We now apply the above framework to prove existence of a limiting process $q_t^\epsilon\rightarrow q_s$ and deriving an SDE satisfied by $q_s$.  Specifically, we have:
\begin{theorem}\label{conv_thm}
Let  $T>0$, $p\geq 2$, $0<\beta\leq p/2$, $x_t^\epsilon=(q_t^\epsilon,p_t^\epsilon)$ satisfy the SDE \req{gen_SDE1}-\req{gen_SDE2}, suppose Assumptions \ref{homog_assump1}-\ref{homog_assump3}  hold, and that the SDE for $q_t^\epsilon$,  \req{gen_SDE1}, can be rewritten in the form
\begin{align}\label{con_thm_eq}
q_t^\epsilon=q_0^\epsilon+\int_0^t\tilde F(s,x_s^\epsilon)ds+\int_0^tG(s,x_s^\epsilon,z_s^\epsilon)ds+\int_0^t\tilde\sigma(s,x_s^\epsilon)dW_s+ R^\epsilon_t
\end{align}
where the components of $G$ have the properties described in Assumption \ref{homog_assump3}, $\tilde F(t,x):[0,\infty)\times\mathbb{R}^{n+m}\rightarrow\mathbb{R}^n$, $\tilde \sigma(t,x):[0,\infty)\times\mathbb{R}^{n+m}\rightarrow\mathbb{R}^{n\times k}$ are continuous, Lipschitz in $x$, uniformly in $t\in[0,T]$, and $R_t^\epsilon$ are continuous semimartingales that satisfy
\begin{align}
E\left[\sup_{t\in[0,T]}\| R_t^\epsilon\|^p\right]=O(\epsilon^\beta)\text{ as }\epsilon\rightarrow 0^+.
\end{align}
Suppose $\tilde G$ (from Assumption \ref{homog_assump3}) is Lipschitz in $x$,  uniformly in $t\in[0,T]$, and we have initial conditions  $E[\|q^\epsilon_0\|^p]<\infty$, $E[\|q_0\|^p]<\infty$, and\\ $E[\|q_0^\epsilon-q_0\|^p]=O(\epsilon^{p/2})$. Then
\begin{align}
E\left[\sup_{t\in[0,T]}\|q_t^\epsilon-q_t\|^p\right]=O(\epsilon^\beta)\text{ as }\epsilon\rightarrow 0^+
\end{align}
where $q_t$ satisfies the SDE
\begin{align}\label{gen_limit_eq}
q_t=q_0+&\int_0^t\tilde F(s,q_s,\psi(s,q_s))ds+\int_0^t\tilde G(s,q_s,\psi(s,q_s))ds\\
&+\int_0^t\tilde\sigma(s,q_s,\psi(s,q_s))dW_s.\notag
\end{align}
\end{theorem}
\begin{proof}
 We will prove this theorem by verifying all the hypotheses of Lemma \ref{conv_lemma}. Define
\begin{align}
\tilde R_t^\epsilon=R_t^\epsilon+\int_0^tG(s,x_s^\epsilon,z_s^\epsilon)ds-\int_0^t\tilde G(s,x_s^\epsilon)ds.
\end{align}
Then
\begin{align}
q_t^\epsilon=q_0^\epsilon+\int_0^t\tilde F(s,x_s^\epsilon)ds+\int_0^t\tilde G(s,x_s^\epsilon)ds+\int_0^t\tilde\sigma(s,x_s^\epsilon)dW_s+ \tilde R^\epsilon_t
\end{align}
where $\tilde F+\tilde G$ and $\tilde \sigma$ are Lipschitz in $x$, uniformly for $t\in[0,T]$ and
\begin{align}
E\left[\sup_{t\in[0,T]}\|\tilde R_t^\epsilon\|^p\right]=O(\epsilon^\beta)
\end{align}
by Theorem \ref{homog_thm}. $E[\|q_0^\epsilon-q_0\|^p]=O(\epsilon^\beta)\text{ as }\epsilon\rightarrow 0^+$ by assumption and 
\begin{equation}
\sup_{t\in[0,T]}E[\|p_t^\epsilon-\psi(t,q_t^\epsilon)\|^p]=O(\epsilon^{p/2})\text{ as }\epsilon\rightarrow 0^+
\end{equation}
by Assumption \ref{homog_assump1}. Note that the assumptions also imply that a solution $q_t$ to \req{gen_limit_eq} exists for all $t\geq 0$  \cite{khasminskii2011stochastic}.

For any $\epsilon>0$, using the  Burkholder-Davis-Gundy inequality and H\"older's inequality we obtain the bound
\begin{align}
&E\left[\sup_{t\in[0,T]}\|q_t^\epsilon\|^p\right]\\
\leq&4^{p-1}\bigg(E\left[\|q_0^\epsilon\|^p\right]+\epsilon^{-p/2}E\left[\left(\int_0^T \|G_1(s,x_s^\epsilon,z_s^\epsilon)\|ds\right)^p\right]\notag\\
&+E\left[\left(\int_0^T \|F_1(s,x_s^\epsilon,z_s^\epsilon)\|ds\right)^p\right]+E\left[\sup_{t\in[0,T]}\left\|\int_0^t\sigma_1(s,x_s^\epsilon)dW_s\right\|^p\right]\bigg)\notag\\
\leq&4^{p-1}\bigg(E\left[\|q_0^\epsilon\|^p\right]+\epsilon^{-p/2}T^{p-1}\int_0^T E\left[\|G_1(s,x_s^\epsilon,z_s^\epsilon)\|^p\right]ds\notag\\
&+T^{p-1}\int_0^T E\left[\|F_1(s,x_s^\epsilon,z_s^\epsilon)\|^p\right]ds+\tilde C T^{p/2-1}\int_0^TE\left[ \|\sigma_1(s,x_s^\epsilon)\|^p_F\right]ds\bigg)\notag.
\end{align}
Polynomial boundedness (see Assumption \ref{homog_assump2}) gives
\begin{align}
&E\left[\sup_{t\in[0,T]}\|q_t^\epsilon\|^p\right]\leq4^{p-1}\bigg(E\left[\|q_0^\epsilon\|^p\right]+\tilde C \int_0^TE\left[ (1+\|z_s^\epsilon\|^q)^p\right]ds\bigg),
\end{align}
where we absorbed all factors of $T$ and $\epsilon$ into the constant $\tilde C$. Using Assumption \ref{homog_assump1} then gives
\begin{align}
&E\left[\sup_{t\in[0,T]}\|q_t^\epsilon\|^p\right]<\infty
\end{align}
for all $\epsilon$ sufficiently small.

Finally, for $n>0$ define the stopping time $\tau_n=\inf\{t\geq 0:\|q_t\|\geq n\}$.  Then for $0\leq t\leq T$ the Lipschitz properties together with the  Burkholder-Davis-Gundy  and H\"older's inequalities imply
\begin{align}
&E\left[\sup_{s\in[0,t]}\|q^{\tau_n}_s\|^p\right]\\
\leq&3^{p-1}\bigg(E[  \|q_0\|^p]+E\left[\left(\int_0^{t\wedge\tau_n} \|(\tilde F+\tilde G)(s,q^{\tau_n}_s,\psi(s,q^{\tau_n}_s))\|ds\right)^p\right]\notag\\
&\hspace{.75cm}+E\left[\sup_{t\in[0,t]}\left\|\int_0^{t\wedge\tau_n}\tilde\sigma(s,q_s^{\tau_n},\psi(s,q^{\tau_n}_s))dW_s\right\|^p\right]\bigg)\notag\\
\leq&3^{p-1}E[  \|q_0\|^p]+\tilde C\int_0^tE\left[\|q^{\tau_n}_s\|^p\right]ds\\
&+\tilde C\int_0^t\|(\tilde F+\tilde G)(s,0,\psi(s,0))\|^p+\|\tilde\sigma(s,0,\psi(s,0))\|^p_Fds\notag\\
\leq&\tilde C\left(1+ \int_0^tE\left[\sup_{r\in[0,s]}\|q^{\tau_n}_r\|^p\right]ds\right),
\end{align}
where $\tilde C$ changes line to line, and is independent of $t$.

The definition of $\tau_n$, together with $E[\|q_0\|^p]<\infty$, implies that  
\begin{equation}
\sup_{s\geq 0}E\left[\sup_{r\in[0,s]}\|q^{\tau_n}_r\|^p\right]<\infty.
\end{equation}
 Therefore we can apply Gronwall's inequality to get
\begin{align}
E\left[\sup_{t\in[0,T]}\|q^{\tau_n}_t\|^p\right]\leq \tilde Ce^{\tilde C T},
\end{align}
where the constant $\tilde C$ is independent of $n$.  Hence, the monotone convergence theorem yields
\begin{align}
E\left[\sup_{t\in[0,T]}\|q_t\|^p\right]\leq \tilde Ce^{\tilde C T}<\infty.
\end{align}

This completes the verification that the hypotheses of Lemma \ref{conv_lemma} hold, allowing us to conclude that
\begin{align}
E\left[\sup_{t\in[0,T]}\|q_t^\epsilon-q_t\|^p\right]=O(\epsilon^\beta)\text{ as }\epsilon\rightarrow 0^+.
\end{align}
\end{proof}

\section{Homogenization of Hamiltonian systems}

In this final section, we apply the above framework to our original Hamiltonian system, \req{Hamiltonian_SDE_q}-\req{Hamiltonian_SDE_p} (in particular, $m=n$ in this section) in order to prove the existence of a limiting process $q_t^\epsilon\to q_t$ and derive a homogenized SDE for $q_t$. Specifically, in Sections \ref{sec:h_explicit_sol} and \ref{sec:chi_explicit_sol} we will study a class of Hamiltonian systems  for which the PDEs  \req{h_eq} for $h$ and \req{chi_eq} for $\chi$ that are needed to derive the limiting equation are explicitly solvable and the required bounds can be verified by elementary means.

The SDE \req{Hamiltonian_SDE_q}-\req{Hamiltonian_SDE_p} can be rewritten in the general form \req{gen_SDE1}-\req{gen_SDE2}:
\begin{align}
dq^\epsilon_t=&\frac{1}{\sqrt{\epsilon}}\nabla_z K(t,q^\epsilon_t,z_t^\epsilon)dt,\label{q_Hamil_eq2}\\
d p^\epsilon_t=&\left(-\frac{1}{\sqrt{\epsilon}}\left(\gamma_l(t,x^\epsilon_t)- \nabla_q\psi_l(t,q_t^\epsilon)\right)\partial_{z_l}K(t,q_t^\epsilon,z_t^\epsilon)-\nabla_q K(t,q^\epsilon_t,z_t^\epsilon)\right.\label{p_Hamil_eq2}\\
&\hspace{5mm}-\nabla_q V(t,q^\epsilon_t)+F(t,x^\epsilon_t)\bigg)dt+\sigma(t,x^\epsilon_t) dW_t,\notag
\end{align}
where $\gamma_l$ denotes the vector obtained by taking the $l$th column of $\gamma$. Specifically, 
\begin{align}
&F_1=0, \hspace{2mm}\sigma_1=0, \hspace{2mm} \sigma_2=\sigma, \hspace{2mm} G_1(t,x,z)=\nabla_z K(t,q,z),\\
&F_2(t,x,z)=-\nabla_q K(t,q,z)-\nabla_q V(t,q)+F(t,x),\\
& G_2(t,x,z)=-\left(\gamma_l(t,x)- \nabla_q\psi_l(t,q)\right)\partial_{z_l}K(t,q,z).
\end{align}
In particular, $\sigma_1=0$, so below we use the definition of $C^{1,2}$ applicable to this case.

The operator $L$, \req{L_def}, and its formal adjoint  have the following form:
\begin{align}
(L\chi)(t,x,z)=&\frac{1}{2} \Sigma_{kl}(t,x)(\partial_{z_k}\partial_{z_l}\chi)(t,x,z)\label{Hamil_L}\\
&-\tilde\gamma_{kl}(t,x)\partial_{z_l}K(t,q,z)(\partial_{z_k}\chi)(t,x,z),\notag\\
(L^*h)(t,x,z)=&\partial_{z_k}\bigg(\frac{1}{2} \Sigma_{kl}(t,x)\partial_{z_l}h(t,x,z)\label{Hamil_L_star}\\
&+\tilde\gamma_{kl}(t,x)\partial_{z_l}K(t,q,z)h(t,x,z)\bigg),\notag
\end{align}
where 
\begin{equation}\label{hamil_Sigma_def}
\Sigma_{ij}=\sum_{\rho}\sigma_{i\rho}\sigma_{j\rho}
\end{equation} and $\tilde\gamma$ was defined in \req{tilde_gamma_def}.  Here $\sigma$ and $\Sigma$ denote what were $\sigma_2$ and $\Sigma_{22}$ respectively in \req{sigma_defs} and $\sigma_1=0$.  In particular, the indices on $\Sigma$ have the meaning $\Sigma_{ij}\equiv (\Sigma_{22})_{ij}$.

\subsection{Computing the noise induced drift}\label{sec:h_explicit_sol}
In general, an explicit solution to $L^*h=0$ is not available, and so the homogenized equation can only be defined  implicitly, as in Theorem \ref{conv_thm}.  However, there are certain classes of systems where we can explicitly derive the form of the additional vector field, $\tilde G$, appearing in the homogenized equation.  In \cite{BirrellHomogenization}, one such class was studied by a different method. Here, we explore the case where the noise and dissipation satisfy the fluctuation dissipation relation pointwise for a time and state dependent generalized temperature $T(t,q)$,
\begin{align}\label{fluc_dis}
\Sigma_{ij}(t,q)=2k_BT(t,q) \gamma_{ij}(t,q).
\end{align}
where $\Sigma$ was defined in \req{hamil_Sigma_def}. We will make Assumptions \ref{assump1}-\ref{assump4}, but make no further constraints on the form of the Hamiltonian here.

As can be verified by a direct calculation, under the assumption \req{fluc_dis}, the adjoint equation \req{Hamil_L_star} is solved by
\begin{align}\label{h_formula}
h(t,q,z)=\frac{1}{Z(t,q)} \exp[-\beta(t,q)K(t,q,z)],
\end{align}
where  we define $\beta(t,q)=1/(k_BT(t,q))$ and $Z$, the ``partition function",  is chosen so that $\int h dz=1$. Note that Assumption \ref{assump3} ensures  such a normalization exists. We also point out that in this case, the antisymmetric part of $\tilde \gamma$ does not contribute to the right hand side of \req{Hamil_L_star}.

An interesting point to note is that when the antisymmetric part of $\tilde\gamma$ vanishes (physically, for $K$ quadratic in $z$ this means a vanishing magnetic field), the vector field that we are taking the divergence of in \req{Hamil_L_star} vanishes identically.    When $\tilde\gamma$ has a non-vanishing antisymmetric part, only once we take the divergence does the expression in \req{Hamil_L_star} vanish.

\begin{comment}
\begin{align}\label{Hamil_L_star}
&(L^*h)(t,x,z)\\
=&\partial_{z_k}\left(\frac{1}{2} \Sigma_{kl}(t,x)\partial_{z_l}h(t,x,z)+\tilde\gamma_{ki}(t,x)\partial_{z_i}K(t,q,z)h(t,x,z)\right)\\
=&\partial_{z_k}\left(-\frac{1}{2} \Sigma_{kl}(t,x)\partial_{z_l}K(t,q,z)/(k_BT(t,q) )h(t,x,z)+\tilde\gamma_{ki}(t,x)\partial_{z_i}K(t,q,z)h(t,x,z)\right)\\
=&\partial_{z_k}\left(-\gamma_{kl}(t,x)\partial_{z_l}K(t,q,z)h(t,x,z)+\tilde\gamma_{ki}(t,x)\partial_{z_i}K(t,q,z)h(t,x,z)\right)\\
=&\partial_{z_k}\left(\tilde\gamma^a_{ki}(t,x)\partial_{z_i}K(t,q,z)h(t,x,z)\right)\\
=&\tilde\gamma^a_{ki}(t,x)\left(\partial_{z_k}\partial_{z_i}K(t,q,z)h(t,x,z)+\partial_{z_i}K(t,q,z)\partial_{z_k}h(t,x,z)\right)\\
=&-\tilde\gamma^a_{ki}(t,x)\partial_{z_i}K(t,q,z)\partial_{z_k}K(t,q,z)/(k_BT(t,q) )h(t,q,z)=0.
\end{align}
\end{comment}

From \req{q_eq}, we see that the terms that require homogenization are
\begin{align}\label{G_gen_def}
G(t,q_t^\epsilon,z_t^\epsilon)=&-(\tilde\gamma^{-1})^{ij}(t,q_t^\epsilon)\partial_{q^j}K(t,q_t^\epsilon,z_t^\epsilon)dt\\
&+(z_t^\epsilon)_j\partial_{q^l}(\tilde\gamma^{-1})^{ij}(t,q_t^\epsilon)\partial_{z_l}K(t,q_t^\epsilon,z_t^\epsilon)dt.\notag
\end{align} 
Using \req{h_formula}, the formal calculation of  Section \ref{sec:formal_G_tilde} gives
\begin{align}
\tilde G(t,q)= -(\tilde\gamma^{-1})^{ij}(t,q)\langle\partial_{q^j}K(t,q,z)\rangle+\frac{\partial_{q^l}(\tilde\gamma^{-1})^{il}(t,q)}{\beta(t,q)},
\end{align}
where we define
\begin{align}\label{h_avg_def}
\langle\partial_{q^j}K(t,q,z)\rangle=\frac{1}{Z(t,q)} \int \partial_{q^j}K(t,q,z)\exp[-\beta(t,q)K(t,q,z)]dz.
\end{align}
\begin{comment}
\begin{align}
\tilde G(t,q)=&\frac{ -(\tilde\gamma^{-1})^{ij}(t,q)}{Z(t,q)} \int \partial_{q^j}K(t,q,z)\exp[-\beta(t,q)K(t,q,z)]dz\\
&+\frac{\partial_{q^l}(\tilde\gamma^{-1})^{ij}(t,q)}{Z(t,q)} \int z_j\partial_{z_l}K(t,q,z) \exp[-\beta(t,q)K(t,q,z)]dz\\
=&\frac{ -(\tilde\gamma^{-1})^{ij}(t,q)}{Z(t,q)} \int \partial_{q^j}K(t,q,z)\exp[-\beta(t,q)K(t,q,z)]dz\\
&-\frac{\partial_{q^l}(\tilde\gamma^{-1})^{ij}(t,q)}{\beta(t,q)Z(t,q)} \int z_j\partial_{z_l} \exp[-\beta(t,q)K(t,q,z)]dz\\
=& -(\tilde\gamma^{-1})^{ij}(t,q)\langle\partial_{q^j}K(t,q,z)\rangle\\
&+\frac{\partial_{q^l}(\tilde\gamma^{-1})^{il}(t,q)}{\beta(t,q)}.
\end{align}
\end{comment}
Of course, this calculation is only formal.  In the next section, we study a particular case where everything can be made rigorous.

\subsection{Rigorous Homogenization of a  class of Hamiltonian systems }\label{sec:chi_explicit_sol}
In this section we explore a class of Hamiltonian systems for which Assumption \ref{homog_assump3} can be rigorously verified via an explicit solution to the PDE for $\chi$. We will work with Hamiltonian systems that satisfy Assumptions \ref{assump1}-\ref{assump5}, \ref{assump7}.  In particular, we are restricting to the class of Hamiltonians with
\begin{align} \label{K_tilde_def}
K(t,q,z)=\tilde K(t,q,A^{ij}(t,q)z_iz_j),
\end{align}
where $A(t,q)$ is valued in the space of positive definite $n \times n$-matrices.  We will write $\tilde K\equiv \tilde K(t,q,\zeta)$ and $\tilde K^\prime\equiv \partial_{\zeta}\tilde K$.

 We will also need the following relations between $\Sigma$, $\gamma$, and $A$  to hold:
\begin{assumption}\label{proportionality_assump}
$\sigma$ is independent of $p$ and
\begin{align}
 \Sigma(t,q)=b_1(t,q)A^{-1}(t,q) , \hspace{2mm} \gamma(t,q)=b_2(t,q)A^{-1}(t,q)
\end{align}
where, for every $T>0$, the $b_i$ are bounded, $C^2$ functions that have positive lower bounds and bounded first derivatives, both on $[0,T]\times\mathbb{R}^n$.
\end{assumption}
Note that these relations imply a fluctuation-dissipation relation with a time and state dependent generalized temperature $T=\frac{b_1}{2k_Bb_2}$.

In  \cite{BirrellHomogenization}, we showed that   Assumptions \ref{assump1}-\ref{assump5} imply:
\begin{align}\label{q_eq2}
d(q_t^\epsilon)^i=&\tilde F^i(t,x)dt+\tilde\sigma^i_{\rho}(t,x_t^\epsilon)dW^\rho_t+G^i(t,x_t^\epsilon,z_t^\epsilon)dt+d(R^\epsilon_t)^i,\notag
\end{align}
where
\begin{align}
\tilde F^i(t,x)=&(\tilde\gamma^{-1})^{ij}(t,q)(-\partial_t\psi_j(t,q)-\partial_{q^j}V(t,q)+F_j(t,x))+S^i(t,q),\\
G^i(t,q,z)=&-(\tilde\gamma^{-1})^{ij}(t,q)(\partial_{q^j}\tilde K)(t,q,A^{ij}(t,q)z_iz_j),\\
 \tilde \sigma^i_\rho(t,x)=& (\tilde\gamma^{-1})^{ij}(t,q)\sigma_{j\rho}(t,x),\\
S^i(t,q)=& k_BT(t,q) \left(\partial_{q^j}(\tilde\gamma^{-1})^{ij}(t,q)-\frac{1}{2}(\tilde\gamma^{-1})^{ik}(t,q)A^{-1}_{jl}(t,q)\partial_{q^k} A^{jl}(t,q)\right),\label{S_def}
\end{align}
and $R_t^\epsilon$ is a family of continuous semimartingales. $S(t,q)$ is called the {\em noise induced drift} (see Eq. 3.26 in \cite{BirrellHomogenization}).  

Note, that with $K(t,q,z)$ defined by \req{K_tilde_def}, the first term in  \req{G_gen_def} consists of two contributions---one coming from the $q$-dependence of $\tilde{K}$ and one coming from the $q$-dependence of $A$.  The $G$ defined here comprises only the first contribution.  The method of \cite{BirrellHomogenization} is able to homogenize the second term in \req{G_gen_def}, as well the first contribution of the first term, leading to the noise induced drift, S, but fails when $\tilde K$ depends explicitly on $q$.  However, under certain circumstances, the method developed in Section \ref{sec:gen_homog} is succeeds in homogenizing the system when $\tilde K$ has $q$ dependence, as we now show.

 We will   need one final assumption:
\begin{assumption}\label{K_poly_bound_assump}
For every $T>0$:
\begin{enumerate}
\item There exists $\zeta_0>0$ and $C>0$ such that $\tilde K^\prime(t,q,\zeta)\geq C$ for all $(t,q,\zeta)\in[0,T]\times\mathbb{R}^n\times[\zeta_0,\infty)$.  
\item $\tilde K(t,q,\zeta)$, $\partial_t\partial_{q^i}\tilde K(t,q,\zeta)$, $\partial_{q^i}\partial_{\zeta}\tilde K(t,q,\zeta)$, and $\partial_{q^i}\partial_{q^j}\tilde K(t,q,\zeta)$ are polynomially bounded in $\zeta$, uniformly in $(t,q)\in[0,T]\times\mathbb{R}^n$.
\end{enumerate}
\end{assumption}

We are now prepared to prove the following homogenization result:
\begin{theorem}\label{hamil_conv_thm}
Let $x_t^\epsilon=(q_t^\epsilon,p_t^\epsilon)$ satisfy the Hamiltonian SDE \req{Hamiltonian_SDE_q}-\req{Hamiltonian_SDE_p} and suppose Assumptions \ref{assump1}-\ref{assump5}, \ref{assump7}, \ref{proportionality_assump}, and \ref{K_poly_bound_assump}  hold. Let $p\geq 2$ and suppose we have initial conditions that satisfy $E[\|q^\epsilon_0\|^p]<\infty$, $E[\|q_0\|^p]<\infty$, and $E[\|q_0^\epsilon-q_0\|^p]=O(\epsilon^{p/2})$. Then for any $T>0$, $0<\beta<p/2$ we have
\begin{align}
E\left[\sup_{t\in[0,T]}\|q_t^\epsilon-q_t\|^p\right]=O(\epsilon^\beta)\text{ as }\epsilon\rightarrow 0^+
\end{align}
where $q_t$ is the solution to the SDE
\begin{align}\label{limit_eq}
dq_t^i=&(\tilde\gamma^{-1})^{ij}(t,q_t)(-\partial_t\psi_j(t,q_t)-\partial_{q^j}V(t,q_t)+F_j(t,q_t,\psi(t,q_t)))dt\\
&+S^i(t,q_t)dt+\tilde G^i(t,q_t)dt+(\tilde\gamma^{-1})^{ij}(t,q_t)\sigma_{j\rho}(t,q_t)dW_t^\rho\notag
\end{align}
with initial condition $q_0$. See \req{tilde_gamma_def}, \req{S_def}, and \req{hamil_tilde_G_def} for the definitions of $\tilde\gamma$, $S$, and $\tilde G$, respectively.
\end{theorem}
\begin{proof}
 From \cite{BirrellHomogenization}, Assumptions \ref{assump1}-\ref{assump5}, \ref{assump7}  imply:
\begin{enumerate}
\item  $q_t^\epsilon$ satisifies an equation of the form \req{con_thm_eq}, where, for every $T>0$, $\tilde F$, $\tilde\sigma$ are bounded, continuous, and Lipschitz in $x$, on $[0,T]\times\mathbb{R}^{2n}$, with Lipschitz constant uniform in $t$. 
\item  Assumptions  \ref{homog_assump1} holds. 
\item For any $p>0$, $T>0$, $0<\beta<p/2$ we have
\begin{align}
E\left[\sup_{t\in[0,T]}\|R_t^\epsilon\|^p\right]=O(\epsilon^\beta) \text{ as } \epsilon \rightarrow 0^+.
\end{align}
\end{enumerate}
Combined with polynomial boundedness of $\tilde K$ (Assumption \ref{K_poly_bound_assump}) we see that Assumption \ref{homog_assump2} also holds. Therefore, to apply  Theorem \ref{conv_thm}, we have to verify Assumption \ref{homog_assump3} and that the $\tilde G$ referenced therein is Lipschitz in $x$, uniformly in $t\in [0,T]$.

From Section \ref{sec:formal_G_tilde}, we expect that
\begin{align}\label{hamil_tilde_G_def}
\tilde G^i(t,q)=-(\tilde\gamma^{-1})^{ij}(t,q)\langle\partial_{q^j}\tilde K(t,q,A^{ij}(t,q)z_iz_j)\rangle
\end{align}
where, similarly to \req{h_avg_def},
\begin{align}
\langle\partial_{q^j}\tilde K(t,q,\|z\|_A^2)\rangle=\frac{1}{Z(t,q)} \int\partial_{q^j}\tilde K(t,q,\|z\|^2_A)\exp[-\beta(t,q)\tilde K(t,q,\|z\|^2_A)]dz.
\end{align}
Here we use the shorthand $\|z\|^2_A\equiv A^{ij}(t,q)z_iz_j$ when the implied values of $t,q$ are apparent from the context.

Using our assumptions, along with several applications of the DCT, one can see that $\tilde G$ is $C^{1}$ and, for every $T>0$, is bounded with bounded first derivatives on $[0,T]\times\mathbb{R}^n$.  In particular, it is Lipschitz in $q$, uniformly in $t\in[0,T]$.

We now turn to solving the equation 
\begin{align}\label{L_pde}
L\chi=G-\tilde G.
\end{align}
  Since $G-\tilde G$ is independent of $p$ and depends on $z$ only through $\|z\|^2_A$, we look for $\chi$ with the same behavior.  Using the ansatz $\chi(t,q,z)=\tilde\chi(t,q,\|z\|^2_A)$, and defining $G^i(t,q,\zeta)=-(\tilde\gamma^{-1})^{ij}(t,q)(\partial_{q^j}\tilde K)(t,q,\zeta)$, leads (on account of the antisymmetry of the matrix $\tilde{\gamma} - \gamma$) to the ODE in the variable $\zeta$:
\begin{align}
&\zeta\tilde\chi^{\prime\prime}(t,q,\zeta) +\left(\frac{n}{2}- \beta(t,q) \zeta  \tilde K^\prime(t,q,\zeta)\right) \tilde\chi^\prime(t,q,\zeta)\\
=&\frac{1}{2 b_1(t,q)}(G(t,q,\zeta)-\tilde G(t,q)).\notag
\end{align}

This has the solution
\begin{align}\label{tilde_chi_def}
&\tilde\chi(t,q,\zeta)=\frac{1}{2b_1(t,q)}\int_0^\zeta \zeta_1^{-n/2}\exp[\beta(t,q) \tilde K(t,q,\zeta_1)]\\
&\times \int_0^{\zeta_1} \zeta_2^{(n-2)/2}\exp[-\beta(t,q) \tilde K(t,q,\zeta_2)]\left(G(t,q,\zeta_2)-\tilde G(t,q)\right)d\zeta_2d\zeta_1\notag
\end{align}

Therefore $\chi(t,q,z)\equiv\tilde\chi(t,q,\|z\|_A^2)$  solves the PDE \req{L_pde}.  One can show that it is is $C^{1,2}$ and that   $\chi$ and its first derivatives are polynomially bounded in $z$, uniformly for $(t,q)\in[0,T]\times\mathbb{R}^{n}$.   As a representative example, in  \ref{app:poly_bound} we outline the proof that $\tilde\chi(t,q,\zeta)$ is polynomially bounded in $\zeta$, uniformly in $(t,q)\in[0,T]\times\mathbb{R}^n$.  The remainder of the computations are similar and we leave them to the reader.

$\chi$ is independent of $p$, so  $\partial_{p_i}\partial_{p_j}\chi=0$ and $\partial_{p_i}\partial_{z_j}\chi=0$.  Therefore, this completes the verification of Assumption \ref{homog_assump3} and we are justified in using  Theorem \ref{conv_thm} to conclude
\begin{align}
E\left[\sup_{t\in[0,T]}\|q_t^\epsilon-q_t\|^p\right]=O(\epsilon^\beta)\text{ as }\epsilon\rightarrow 0^+,
\end{align}
where $q_t$ satisfies the SDE
\begin{align}
q_t=q_0+&\int_0^t\tilde F(s,q_s,\psi(s,q_s))ds+\int_0^t\tilde G(s,q_s)ds+\int_0^t\tilde\sigma(s,q_s)dW_s\notag 
\end{align}
as claimed.

\end{proof}

Lastly, we give an example of a general class of Hamiltonians that satisfy the hypotheses of Theorem \ref{hamil_conv_thm}.  The proof of this corollary is straighforward, so we leave it to the reader.
\begin{corollary}
Consider the class of Hamiltonians of the form
\begin{align}
H(t,q,p)=\sum_{l=k_1}^{k_2} d_l(t,q) \left[A^{ij}(t,q) (p-\psi(t,q))_i(p-\psi(t,q))_j\right]^l+V(t,q)
\end{align}
where $1\leq k_1\leq k_2$ are integers  and the following properties hold on $[0,T]\times\mathbb{R}^{n}$ for every $T>0$:
\begin{enumerate}
\item $V$ is $C^2$ and $\nabla_q V$ is bounded and Lipschitz in $q$, uniformly in $t\in[0,T]$.
\item $\psi$ is $C^3$ and $\partial_t\psi$, $\partial_{q^i}\psi$, $\partial_t\partial_{q^i}\psi$,  $\partial_{q^i}\partial_{q^j}\psi$, $\partial_t\partial_{q^j}\partial_{q^i}\psi$, and $\partial_{q^l}\partial_{q^j}\partial_{q^i}\psi$ are bounded.
\item $d_l$ are $C^2$, non-negative, bounded, and have bounded first  and second derivatives.
\item $d_{k_1}$ and $d_{k_2}$ are uniformly bounded below by a positive constant.
\item $A$ is $C^2$, positive-definite, and $A$, $\partial_t A$, $\partial_{q_i} A$,  $\partial_t \partial_{q^i}A$,  and $\partial_{q^i}\partial_{q^j} A$  are bounded.
\item The eigenvalues of $A$ are uniformly bounded below by a positive constant.
\end{enumerate}
Also suppose that
\begin{enumerate}
\item $\sigma$ is independent of $p$ and
\begin{align}
 \Sigma(t,q)=b_1(t,q)A^{-1}(t,q) , \hspace{2mm} \gamma(t,q)=b_2(t,q)A^{-1}(t,q)
\end{align}
where, for every $T>0$, the $b_i$ are bounded, $C^2$ functions with positive lower bounds and bounded first derivatives.
\item $\gamma$ is $C^2$, is independent of $p$, and  $\partial_t\gamma$, $\partial_{q^i} \gamma$, $\partial_t\partial_{q^j}\gamma$,  $\partial_{q^i}\partial_{q^j}\gamma$ are bounded on $[0,T]\times\mathbb{R}^{n}$.
 \item The eigenvalues of $\gamma$ are bounded below by some $\lambda>0$.
\item  $\gamma$, $F$, and $\sigma$ are bounded.
\item  $F$ and $\sigma$ are Lipschitz in $x$ uniformly in $t\in[0,T]$.
\item There exists $C>0$ such that the (random) initial conditions satisfy $K^\epsilon(0,x^\epsilon_0)\leq C$ for all $\epsilon>0$ and all $\omega\in\Omega$.
\item  There is a $p\geq 2$ such that 
\begin{align}
E[\|q^\epsilon_0\|^p]<\infty, \hspace{2mm} E[\|q_0\|^p]<\infty,\text{ and } E[\|q_0^\epsilon-q_0\|^p]=O(\epsilon^{p/2}).
\end{align}
\end{enumerate}
 Then  all the hypotheses of Theorem \ref{hamil_conv_thm} hold, in particular Assumptions \ref{assump1}-\ref{assump5}, \ref{assump7}, \ref{proportionality_assump}, and \ref{K_poly_bound_assump}  hold, and hence, for any $\beta \in \left(0, {p \over 2}\right)$, 
\begin{align}
E\left[\sup_{t\in[0,T]}\|q_t^\epsilon-q_t\|^p\right]=O(\epsilon^\beta)\text{ as }\epsilon\rightarrow 0^+,
\end{align}
where $x_t^\epsilon=(q_t^\epsilon,p_t^\epsilon)$ satisfy the Hamiltonian SDE \req{Hamiltonian_SDE_q}-\req{Hamiltonian_SDE_p} and $q_t$ satisfies  the homogenized SDE, \req{limit_eq}.
\end{corollary}

\begin{comment}
Checked:\\
Assump1\\
Assump2\\
Assump3\\
Assump4\\
Assump5\\
Assump7\\
\ref{proportionality_assump}\\
 \ref{K_poly_bound_assump}\\

\begin{assumption}\label{K_poly_bound_assump}
For every $T>0$:
\begin{enumerate}
\item There exists $\zeta_0>0$ and $C>0$ such that $\tilde K^\prime(t,q,\zeta)\geq C$ for all $(t,q,\zeta)\in[0,T]\times\mathbb{R}^n\times[\zeta_0,\infty)$.  
\item $\tilde K(t,q,\zeta)$, $\partial_t\partial_{q^i}\tilde K(t,q,\zeta)$, $\partial_{q^i}\partial_{\zeta}\tilde K(t,q,\zeta)$, and $\partial_{q^i}\partial_{q^j}\tilde K(t,q,\zeta)$ are polynomially bounded in $\zeta$, uniformly in $(t,q)\in[0,T]\times\mathbb{R}^n$.
\end{enumerate}

\begin{align}
\tilde K(t,q,\zeta)=\sum_{l=k_1}^{k_2} d_l(t,q)\zeta^l
\end{align}
is poly bounded  in $\zeta$, uniformly in $(t,q)\in[0,T]\times\mathbb{R}^n$ since the $d_l$ have bounded first and second derivatives

\begin{align}
\tilde K^\prime(t,q,\zeta)=\sum_{l=k_1}^{k_2} ld_l(t,q)\zeta^{l-1}\geq k_1 d_{k_1}(t,q)\zeta^{l-1}\geq \tilde C\zeta^{l-1}\geq C
\end{align}
for $\zeta\geq 1$.
since $d_l$ are non-neg., $d_{k_1}$ is bounded below.  This proves \ref{K_poly_bound_assump}.
\end{assumption}

\end{comment}

\appendix

\numberwithin{assumption}{section}
\section{Material from \cite{BirrellHomogenization}}\label{app:assump}
\setcounter{assumption}{0}
    \renewcommand{\theassumption}{\Alph{section}\arabic{assumption}}

    \renewcommand{\thelemma}{\Alph{section}\arabic{lemma}}

In this appendix, we collect several useful assumptions and results from \cite{BirrellHomogenization}. The assumptions listed here are {\em not} used in the entirety of  this current work. When they are needed for a particular result we explicitly references them.

\begin{assumption}\label{assump1}
We assume that the Hamiltonian has the form given in \req{H_form} where $K$ and $\psi$ are $C^2$ and $K$ is non-negative.  For every $T>0$, we assume the following bounds hold on $[0,T]\times\mathbb{R}^{2n}$:
\begin{enumerate}
\item There exist $C>0$ and $M>0$ such that
\begin{align}\label{K_assump1}
\max\{|\partial_t K(t,q,z)|,\|\nabla_qK(t,q,z)\|\}\leq M+CK(t,q,z).
\end{align}
\item There exist  $c>0$ and $M\geq 0$ such that
\begin{align}\label{K_assump2}
\|\nabla_z K(t,q,z)\|^2+M\geq c K(t,q,z).
\end{align}
\item
For every $\delta>0$ there exists an $M>0$ such that
\begin{align}\label{K_assump3}
\max\left\{\|\nabla_z K(t,q,z)\|,\left(\sum_{ij}|\partial_{z_i}\partial_{z_j}K(t,q,z)|^2\right)^{1/2}\right\}\leq M+\delta K(t,q,z).
\end{align}
\end{enumerate}

\end{assumption}

\begin{assumption}\label{assump2}
For every $T>0$, we assume that the following hold uniformly on $[0,T]\times\mathbb{R}^n$:
\begin{enumerate}
\item $V$ is $C^2$ and $\nabla_q V$ is bounded
\item $\gamma$ is symmetric with eigenvalues bounded below by some $\lambda>0$.
\item  $\gamma$, $F$, $\partial_t\psi$, and $\sigma$ are bounded.
\item There exists $C>0$ such that the (random) initial conditions satisfy $K^\epsilon(0,x^\epsilon_0)\leq C$ for all $\epsilon>0$ and all $\omega\in\Omega$.
\end{enumerate}
\end{assumption}

\begin{assumption}\label{assump3}
We assume that for every $T>0$ there exists $c>0$, $\eta>0$ such that
\begin{align}
K(t,q,z)\geq c\|z\|^{2\eta}
\end{align}
on $[0,T]\times\mathbb{R}^{2n}$.
\end{assumption}

\begin{assumption}\label{assump4}
We assume that $\gamma$ is $C^1$ and is independent of $p$.
\end{assumption}

\begin{assumption}\label{assump5}
We assume that $K$ has the form
\begin{align}
K(t,q,z)=\tilde K(t,q,A^{ij}(t,q)z_iz_j)
\end{align}
where $\tilde K(t,q,\zeta)$ is $C^2$ and non-negative on $[0,\infty)\times\mathbb{R}^n\times[0,\infty)$ and $A(t,q)$ is a $C^2$ function whose values are symmetric $n \times n$-matrices.   We also assume that for every $T>0$, the eigenvalues of $A$ are bounded above and below by some constants $C>0$ and $c>0$ respectively, uniformly on $[0,T]\times\mathbb{R}^n$.  

We will write $\tilde K^\prime$ for $\partial_\zeta \tilde K$ and will use the abbreviation $\|z\|_A^2$ for $A^{ij}(t,q)z_iz_j$ when the implied values of $t$ and $q$ are apparent from the context.
\end{assumption}
\setcounter{assumption}{6}
\begin{assumption}\label{assump7}
We assume that, for every $T>0$, $\nabla_q V$, $F$, and $\sigma$ are Lipschitz in $x$ uniformly in $t\in[0,T]$.  We also assume that $A$ and $\gamma$ are $C^2$, $\psi$ is $C^3$, and $\partial_t\psi$, $\partial_{q^i}\psi$, $\partial_{q^i}\partial_{q_j}\psi$, $\partial_t\partial_{q^i}\psi$, $\partial_t\partial_{q^j}\partial_{q^i}\psi$, $\partial_{q^l}\partial_{q^j}\partial_{q^i}\psi$, $\partial_t\gamma$, $\partial_{q^i} \gamma$, $\partial_t\partial_{q^j}\gamma$,  $\partial_{q^i}\partial_{q^j}\gamma$, $\partial_t A$, $\partial_{q^i} A$, $\partial_t \partial_{q^i}A$,  and $\partial_{q^i}\partial_{q^j} A$ are bounded on $[0,T]\times\mathbb{R}^{2n}$ for every $T>0$.
\end{assumption}
Note that, combined with Assumptions \ref{assump1}-\ref{assump4}, this implies $\tilde\gamma$, $\tilde\gamma^{-1}$, $\partial_t\tilde\gamma^{-1}$, $\partial_{q^i}\tilde\gamma^{-1}$, $\partial_t\partial_{q^j}\tilde\gamma^{-1}$, and  $\partial_{q^i}\partial_{q^j}\tilde\gamma^{-1}$ are bounded on compact $t$-intervals.

\begin{lemma}\label{q_bounded_lemma}
Under Assumptions \ref{assump1} and \ref{assump2}, for any $T>0$, $p>0$ we have
\begin{align}
E\left[\sup_{t\in[0,T]}\|q_t^\epsilon\|^p\right]<\infty.
\end{align}
\end{lemma}

\begin{lemma}\label{p_decay_lemma}
Under Assumptions \ref{assump1}-\ref{assump3}, for any $T>0$, $p>0$ we have
\begin{align}
\sup_{t\in[0,T]}E[\|p_t^\epsilon-\psi(t,q_t^\epsilon)\|^{p}]=O(\epsilon^{p/2}) \text{ as }\epsilon\rightarrow 0^+
\end{align}
and for any $p>0$,  $T>0$, $0<\beta<p/2$ we have
\begin{align}
E\left[\sup_{t\in[0,T]}\|p_t^\epsilon-\psi(t,q_t^\epsilon)\|^{p}\right]=O(\epsilon^\beta) \text{ as }\epsilon\rightarrow 0^+.
\end{align}
\end{lemma}

The following is a slight variant of the result from \cite{BirrellHomogenization}, but the proof is identical.
\begin{lemma}\label{conv_lemma}
Let $T>0$ and suppose we have continuous functions $\tilde F(t,x):[0,\infty)\times\mathbb{R}^{n+m}\rightarrow\mathbb{R}^n$, $\tilde \sigma(t,x):[0,\infty)\times\mathbb{R}^{n+m}\rightarrow\mathbb{R}^{n\times k}$, and  $\psi:[0,\infty)\times\mathbb{R}^n\rightarrow\mathbb{R}^m$ that are Lipschitz in $x$, uniformly in $t\in[0,T]$.

Let $W_t$ be a $k$-dimensional Wiener process, $p\geq 2$ and $\beta>0$ and suppose that we have continuous semimartingales $q_t$ and, for each $0<\epsilon\leq\epsilon_0$, $\tilde R_t^\epsilon$, $x_t^\epsilon=(q_t^\epsilon,p_t^\epsilon)$ that satisfy the following properties:
\begin{enumerate}
\item $q_t^\epsilon=q_0^\epsilon+\int_0^t\tilde F(s,x_s^\epsilon)ds+\int_0^t\tilde\sigma(s,x_s^\epsilon)dW_s+\tilde R^\epsilon_t$.\label{q_eps_eq_assump}
\item $q_t=q_0+\int_0^t\tilde F(s,q_s,\psi(s,q_s))ds+\int_0^t\tilde\sigma(s,q_s^\epsilon,\psi(s,q_s))dW_s$.\label{q_eq_assump}
\item $E[\|q_0^\epsilon-q_0\|^p]=O(\epsilon^\beta)\text{ as }\epsilon\rightarrow 0^+$. \label{ics_assump}
\item $E\left[\sup_{t\in[0,T]}\|\tilde R_t^\epsilon\|^p\right]=O(\epsilon^\beta)\text{ as }\epsilon\rightarrow 0^+$.\label{R_decay_assump}
\item $\sup_{t\in[0,T]}E[\|p_t^\epsilon-\psi(t,q_t^\epsilon)\|^p]=O(\epsilon^\beta)\text{ as }\epsilon\rightarrow 0^+$.\label{p_decay_assump}
\item $E\left[\sup_{t\in[0,T]}\|q_t^\epsilon\|^p\right]<\infty$ for all $\epsilon>0$ sufficiently small.\label{q_eps_integrable_assump}
\item $E\left[\sup_{t\in[0,T]}\|q_t\|^p\right]<\infty$.\label{q_integrable_assump}
\end{enumerate}

Then 
\begin{align}
E\left[\sup_{t\in[0,T]}\|q_t^\epsilon-q_t\|^p\right]=O(\epsilon^\beta)\text{ as }\epsilon\rightarrow 0^+.
\end{align}

\end{lemma}

\section{Polynomial boundedness of $\tilde\chi$}\label{app:poly_bound}
Changing variables, $\tilde\chi$ can be expressed as
\begin{align}
\tilde\chi(t,q,\zeta)=&\frac{1}{2b_1(t,q)}\zeta\int_0^1\exp[\beta(t,q) \tilde K(t,q,s\zeta)]\\
&\times \int_0^{1} r^{(m-2)/2}\exp[-\beta(t,q) \tilde K(t,q,rs\zeta)]\left(G(t,q,rs\zeta)-\tilde G(t,q)\right)dr ds.\notag
\end{align}
Applying the DCT to this expression several times, one can prove that $\tilde\chi$ is $C^{1,2}$. Using the fact that $\tilde K$ and $\partial_{q^i}\tilde K$ are polynomially bounded in $\zeta$, uniformly in $(t,q)\in[0,T]\times\mathbb{R}^n$, one can see that $\tilde\chi(t,q,\zeta)$ is bounded on $[0,T]\times\mathbb{R}^n\times[0,\zeta_0]$ for any $\zeta_0>0$.  From Assumption \ref{K_poly_bound_assump}, there exists $\zeta_0$ and $C>0$ such that $\tilde K^\prime(t,q,\zeta)\geq C$ for all $(t,q,\zeta)\in[0,T]\times\mathbb{R}^n\times[\zeta_0,\infty)$.  

By combining \req{tilde_chi_def} with \req{hamil_tilde_G_def}, one finds that for $\zeta\geq \zeta_0$, $\tilde \chi$ can alternatively be written as
\begin{align}
\tilde\chi(t,q,\zeta)=&\tilde\chi(t,q,\zeta_0)+\frac{1}{2b_1(t,q)}\int_{\zeta_0}^\zeta \zeta_1^{-m/2}\exp[\beta(t,q) \tilde K(t,q,\zeta_1)]\\
&\times \int_{\zeta_1}^\infty \exp[-\beta(t,q)\tilde K(t,q,\zeta_2)]\zeta_2^{(m-2)/2}(\tilde G(t,q)-G(t,q,\zeta_2))d\zeta_2 d\zeta_1\notag.
\end{align}
Therefore, if we can show that the second term has the polynomial boundedness property then so does $\tilde\chi$, and hence $\chi$.

Letting $\tilde C$ denote a constant that potentially changes in each line and choosing $\zeta_0$ as in Assumption \ref{K_poly_bound_assump}, we have
\begin{align}
&\bigg\|\frac{1}{2b_1(t,q)}\int_{\zeta_0}^\zeta \zeta_1^{-m/2}\exp[\beta(t,q) \tilde K(t,q,\zeta_1)]\\
&\times \int_{\zeta_1}^\infty \exp[-\beta(t,q)\tilde K(t,q,\zeta_2)]\zeta_2^{(m-2)/2}(\tilde G(t,q)-G(t,q,\zeta_2))d\zeta_2 d\zeta_1\bigg\|\notag\\
\leq &\tilde C\int_{\zeta_0}^\zeta \zeta_1^{-m/2}\exp[\beta(t,q) \tilde K(t,q,\zeta_1)]\notag\\
&\times \int_{\zeta_1}^\infty \exp[-\beta(t,q)\tilde K(t,q,\zeta_2)]\zeta_2^{(m-2)/2+q} d\zeta_2 d\zeta_1\notag\\
\leq &\tilde C \zeta_0^{-m/2}\int_{\zeta_0}^\zeta \exp[\beta(t,q) \tilde K(t,q,\zeta_1)]\notag\\
&\times \int_{\zeta_1}^\infty \exp[-\beta(t,q)\tilde K(t,q,\zeta_2)] \tilde K(t,q,\zeta_2)^{((m-2)/2+q)/\eta} \tilde K^\prime(t,q,\zeta_2) d\zeta_2 d\zeta_1\notag\\
=&\tilde C \zeta_0^{-m/2}\int_{\zeta_0}^\zeta \exp[\beta(t,q) \tilde K(t,q,\zeta_1)]\notag\\
&\hspace{1.5cm}\times\int_{\tilde K(t,q,\zeta_1)}^\infty \exp[-\beta(t,q)u]u^{((m-2)/2+q)/\eta} dud\zeta_1\notag
\end{align}
for some $q>0$. To obtain the first inequality, we use polynomial boundedness of $\partial_{q^i}\tilde K$.  For the second, we used Assumption \req{assump3} together with the fact that $\tilde K^\prime\geq C>0$ on $[0,T]\times\mathbb{R}^n\times[\zeta_0,\infty)$.

Therefore we obtain
\begin{align}
\|\tilde\chi(t,q,\zeta)\|\leq\tilde C\left(1+\int_{\zeta_0}^\zeta P(\tilde K(t,q,\zeta_1))d\zeta_1\right)
\end{align}
for some  polynomial $P(x)$ with positive coefficients that are independant of $t$ and $q$. Polynomial boundedness of $\tilde K$ then implies
\begin{align}
\|\tilde\chi(t,q,\zeta)\|\leq\tilde C\left(1+ \int_{\zeta_0}^\zeta Q(\zeta_1)d\zeta_1\right)
\end{align}
for some polynomial $Q(\zeta)$.  This proves the desired polynomial boundedness  property for $\tilde\chi$.

\subsection*{Acknowledgments}

J.W. was partially supported by NSF grants DMS 131271 and DMS 1615045.\\

\bibliographystyle{ieeetr}

\begin{thebibliography}{1}

\bibitem{khasminskii2011stochastic}
R.~Khasminskii, {\em Stochastic Stability of Differential Equations}, vol.~66, Springer Science \& Business Media, 2011.

\bibitem{arnold}
L.~Arnold, {\em Stochastic Differential Equations: Theory and Applications}, \newblock Krieger Pub Co, 1992.

\bibitem{ortega2013matrix}
J.M.~Ortega,  {\em Matrix Theory: A Second Course}, Springer, 2013.

\bibitem{Zwanzig}

R. Zwanzig, {\em Nonequilibrium statistical mechanics}, Oxford University Press, 2001.




\bibitem{Hottovy}

S. ~Hottovy, A. ~McDaniel, G. ~Volpe and J. ~Wehr, {\em The Smoluchowski-Kramers limit of stochastic differential equations with arbitrary state-dependent friction}, Commun. Math. Phys. (2015) 336: 1259. 



\bibitem{clt}
J. ~Birrell, J. ~Wehr, \newblock Phase space homogenization of noisy {H}amiltonian systems.
\newblock {\em Annales Henri Poincar{\'e}}, 19(4):1081--1114, Apr 2018.


\bibitem{karatzas2014brownian}
I.~Karatzas and S.~Shreve, {\em Brownian Motion and Stochastic Calculus}, Graduate Texts in Mathematics, Springer New York, 2014.



\bibitem{BirrellHomogenization}
J.~{Birrell} and J.~{Wehr}, 
\newblock Homogenization of dissipative, noisy, {H}amiltonian dynamics.
\newblock {\em Stochastic Processes and their Applications}, 2017.



\bibitem{pavliotis2008multiscale}
G.~Pavliotis and A.~Stuart, {\em Multiscale Methods: Averaging and Homogenization}, Texts in Applied Mathematics, Springer New York, 2008.








\end{thebibliography}

\end{document}